\pdfoutput=1
\documentclass[letterpaper, 10 pt, conference]{ieeeconf}  

\usepackage{float}
\usepackage{dblfloatfix}
\usepackage{amssymb,latexsym,amsmath,amsthm,graphicx,epstopdf,subfig}
\usepackage[hidelinks]{hyperref}

\theoremstyle{plain}
\newtheorem{theorem}{Theorem}
\newtheorem{lemma}{Lemma}

\theoremstyle{definition}

\newtheorem{example}[theorem]{Example}

\theoremstyle{remark}
\newtheorem{remark}{Remark}
\newtheorem{assumption}{Assumption}

\IEEEoverridecommandlockouts                              

\title{\LARGE \bf
A Control Theoretical Adaptive Human Pilot Model: Theory and Experimental Validation}

\author{S. S. Tohidi$^{1}$ and Y. Yildiz$^{1}$\\
\thanks{$^{1}$S. S. Tohidi and Y. Yildiz are with Faculty of Mechanical Engineering, Bilkent University, Cankaya, Ankara 06800, Turkey
        {\tt\small \{shahabaldin, yyildiz\}@bilkent.edu.tr}}%
}

\begin{document}

\maketitle
\thispagestyle{empty}
\pagestyle{empty}

\begin{abstract}


This paper proposes an adaptive human pilot model that is able to mimic the crossover model in the presence of uncertainties. The proposed structure is based on the model reference adaptive control, and the adaptive laws are obtained using the Lyapunov-Krasovskii stability criteria. The model can be employed for human-in-the-loop stability and performance analyses incorporating different types of controllers and plant types. For validation purposes, an experimental setup is employed to collect data and a statistical analysis is conducted to measure the predictive power of the pilot model. 
\end{abstract}

\section{INTRODUCTION}

Humans' unique abilities such as adaptive behavior in dynamic environments, and social interaction and moral judgment capabilities, make them essential elements of many control loops. On the other hand, compared to humans, automation provides higher computational performance and multi-tasking capabilities without any fatigue, stress, or boredom \cite{Noth16,Kor15}. Although they have their own individual strengths, humans and automation also demonstrate several weaknesses. Humans may have anxiety, fear and may become unconscious during an operation. Furthermore, in the tasks that require increased attention and focus, humans tend to provide high gain control inputs that can cause undesired oscillations. One example of this phenomenon, for example, is the occurrence of pilot induced oscillations (PIO), where undesired and sustained oscillations are observed due to an abnormal coupling between the aircraft and the pilot \cite{YilKol10, Yil11a,AcoYil14,TohYil18}. Similarly, automation may fail due to an uncertainty, fault or cyber-attack \cite{Li14}. Thus, it is more preferable to design systems where humans and automation work in harmony, complementing each other, resulting in a structure that benefits from the advantages of both.

To achieve a reliable human-automation harmony, a mathematically rigorous human operator model is paramount. A human operator model helps develop safe control systems, and provide a better prediction of human actions and limitations \cite{Hul13,Yuc18,Emre19, Zha19}. 
Quasi-linear model \cite{Mcr57} is one of the first human operator models, which consists of a describing function and a remnant signal accounting for nonlinear behavior. An overview of this model is provided in \cite{Mcr74}. In some applications, where the linear behavior may be dominant, the nonlinear part of this model can be ignored, and the resulting lead-lag-type compensator is used in closed loop stability analysis \cite{Neal71}. The crossover model, proposed in \cite{Mcr63}, is another important human operator model in the aerospace domain. It is motivated from the empirical observations that human pilots adapt their responses in such a way that the overall system dynamics resembles that of a well designed feedback system \cite{Bee08}. A generalized crossover model which mimics human behavior when controlling a fractional order plant is proposed in \cite{Mar17}. In \cite{War16}, crossover model is employed to provide information about the human intent for the controller. In \cite{Gil09}, the dynamics of the operator is represented as a spring-damper-mass system.

Control theoretical operator models drawing from the optimal and adaptive control theories are also proposed by several authors. 
Optimal human models are based on the idea that a well trained human operator behaves in an optimal manner \cite{Wier69, Kle70, Na12, Lon13, Hu19}. On the other hand, adaptive models, such as the ones proposed in \cite{Hess09, Hess15} and \cite{TohYil19Human}, aim to replicate the adaptation capability of humans in uncertain and dynamics environments. In \cite{Hess09} and \cite{Hess15}, adaptation rules are proposed based on expert knowledge. The adaptive model proposed in \cite{Hess09} is applied to change the parameters of the pilot model based on force feedback from a smart inceptor \cite{Xu19}. 
A survey on various pilot models can be found in \cite{Lon14} and \cite{Xu17}.



 Several approaches are also developed for human model parameter identification. In \cite{Zaal11}, a two-step method using wavelets and a windowed maximum likelihood estimation is exploited for the estimation of time-varying pilot model parameters. In \cite{Duar17}, a linear parameter varying model identification framework is incorporated to estimate time-varying human state space representation matrices. Subsystem identification is used in \cite{Zha15} to model human control strategies. In \cite{Van15}, a human operator model for preview tracking tasks is derived from measurement data.


In this paper, we build upon the earlier successful pilot models and propose an adaptive human pilot model that modifies its behavior based on plant uncertainties. This model distinguishes itself from earlier adaptive models by having mathematically derived laws to achieve a cross-over-model-like behavior, instead of employing expert knowledge. This allows a rigorous stability proof, using the Lyapunov-Krasovskii stability criteria, of the overall closed loop system. To validate the model, a setup, including a joystick and a monitor, is used. The participant data collected through this experimental setup is subjected to visual and statistical analyses to evaluate the accuracy of the proposed model. Initial research results of this study were presented in \cite{TohYil19Human}, where the details of the mathematical proof and human experimental validation studies were missing.


This paper is organized as follows. In Section \ref{sec:prob}, the problem statement is given. Obtaining reference model parameters, which determine the properties of the cross-over model, is discussed in Section \ref{sec:reference}. Section \ref{sec:human} presents the human control strategy together with a stability analysis. Experimental set-up, results, and a statistical analysis are provided in Section \ref{sec:experiment}. Finally, a summary is given in Section \ref{sec:conclusion}.

\section{PROBLEM STATEMENT}\label{sec:prob}

According to McRuer's crossover model \cite{Mcr67}, human pilots in the control loop behave in a way that results in an open loop transfer function
\begin{equation}\label{eq:e1x}
Y_{OL}(s)=Y_{h}(s)Y_{p}(s)=\frac{\omega_ce^{-\tau s}}{s},
\end{equation}
near the crossover frequency, $ \omega_c $, where $ Y_{h} $ is the transfer function of the human pilot and $ Y_{p} $ is the transfer function of the plant. $ \tau $ is the effective time delay, including transport delays and high frequency neuromuscular lags.

Consider the following plant dynamics
\begin{equation}\label{eq:e2x}
\dot{x}_p(t)=A_px_p(t)+B_p u_p(t),
\end{equation}
where $ x_p\in \mathbb{R}^{n_p} $ is the plant state vector, $ u_p\in \mathbb{R} $ is the input vector, $ A_p\in \mathbb{R}^{n_p\times n_p} $ is an unknown state matrix and $ B_p\in \mathbb{R}^{n_p} $ is an unknown input matrix.

The human \textit{neuromuscular model} \cite{Mag71, Van04} is represented in state space form as
\begin{equation}\label{eq:e3x}
\begin{aligned}
\dot{x}_h(t)&=A_hx_h(t)+B_hu(t-\tau) \\
y_h(t)&=C_hx_h(t)+D_hu(t-\tau),
\end{aligned} 
\end{equation}
where $ x_h\in \mathbb{R}^{n_h} $ is the neuromuscular state vector, $ A_h\in \mathbb{R}^{n_h\times n_h} $ is the state matrix, $ B_h\in \mathbb{R}^{n_h} $ is the input matrix, $ C_h\in \mathbb{R}^{1\times n_h} $ is the output matrix and $ D_h\in \mathbb{R} $ is the control output matrix. $ u\in \mathbb{R} $ is the neuromuscular input vector, which represents the control decisions taken by the human and fed to the neuromuscular system, $ y_h\in \mathbb{R} $ is the output vector, and $ \tau\in \mathbb{R}^+ $ is a known, constant delay. The neuromuscular model parameters are assumed to be known and the output of the model, $ y_h $, is used as the plant input $ u_p $ in (\ref{eq:e2x}), that is $ y_h=u_p $ (see Fig. \ref{fig:f1-2}). 
\begin{figure}
	\vspace{0.3cm}
	\begin{center}
		\includegraphics[width=9.0cm]{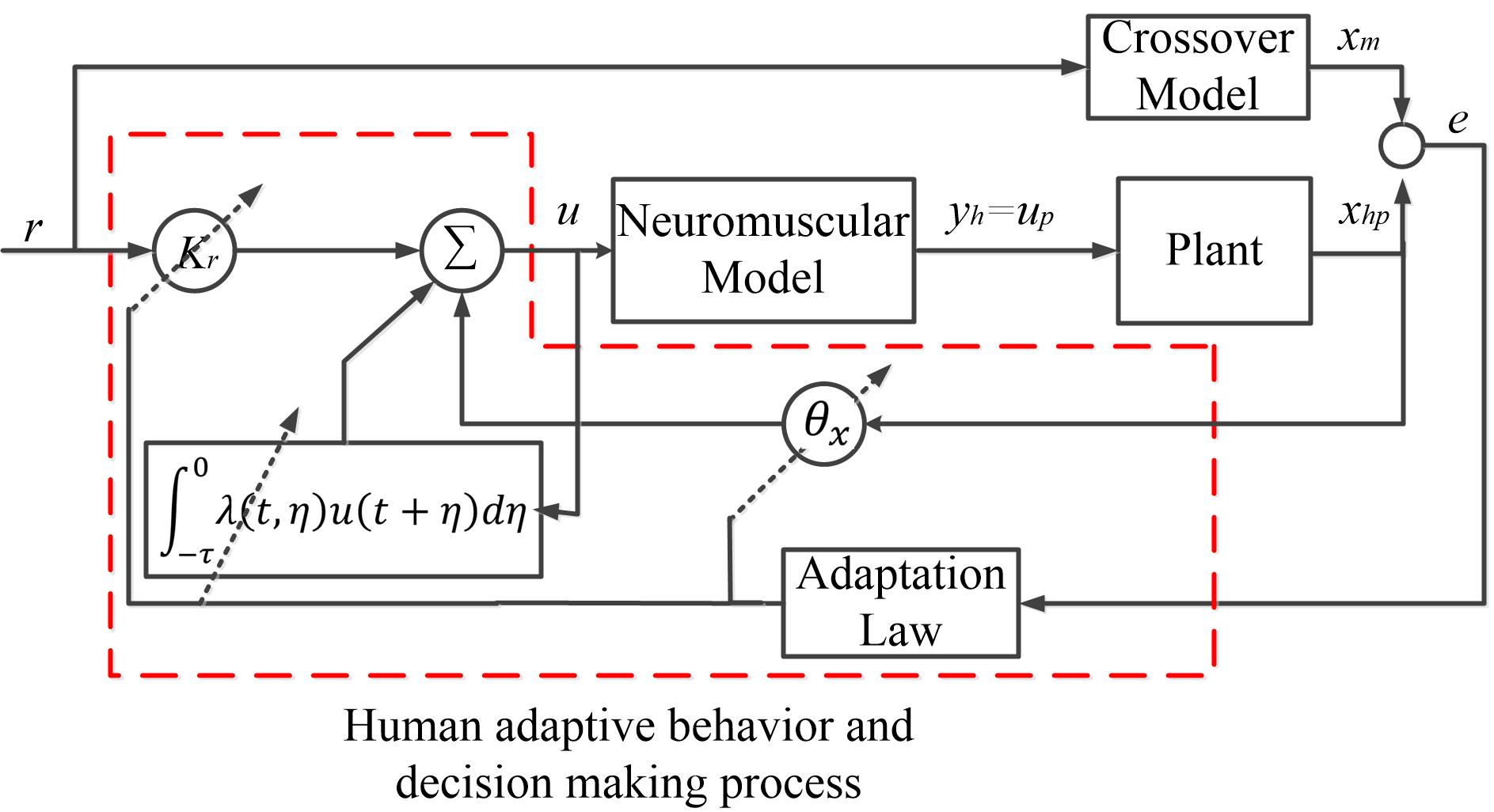}    %
		\caption{The block diagram of the human adaptive behavior and decision making in a closed loop system.} 
		\label{fig:f1-2}
	\end{center}
\end{figure}


By aggregating the human pilot and plant states, we obtain the combined open loop human neuromuscular and plant dynamics as
\begin{equation}\label{eq:e4x}
\begin{aligned}
\underbrace{\begin{bmatrix}
	\dot{x}_h(t) \\ \dot{x}_p(t)
	\end{bmatrix}}_{\dot{x}_{hp}(t)}&=
\underbrace{\begin{bmatrix}
	A_h & 0_{n_h\times n_p} \\ B_p C_h & A_p
	\end{bmatrix}}_{A_{hp}}\underbrace{\begin{bmatrix}
	x_h(t) \\ x_p(t)
	\end{bmatrix}}_{x_{hp(t)}}\\ &+
\underbrace{\begin{bmatrix}
	B_h \\ B_pD_h
	\end{bmatrix}}_{B_{hp}}u(t-\tau), 
\end{aligned} 
\end{equation}
which can be written in the following compact form
\begin{equation}\label{eq:e5x}
\begin{aligned}
\dot{x}_{hp}(t)&=A_{hp}x_{hp}(t)+B_{hp}u(t-\tau),
\end{aligned} 
\end{equation}
where $ x_{hp}=[x_h^T\ x_p^T]^T\in \mathbb{R}^{(n_p+n_h)} $, $ A_{hp}\in \mathbb{R}^{(n_p+n_h)\times (n_p+n_h)} $, $ B_{hp}\in \mathbb{R}^{(n_p+n_h)} $. 
\begin{assumption}
	The pair $ (A_{hp}, B_{hp}) $ is controllable.
\end{assumption}

The goal is to obtain the input $ u(t) $ in (\ref{eq:e3x}), which is the human pilot control decision variable, such that the closed loop system consisting of the adaptive human pilot model and the plant follow the output of a unity feedback reference model with an open loop crossover model transfer function. The closed loop transfer function of the reference model is therefore calculated as
\begin{equation}\label{eq:e21xqx}
\begin{aligned}
G_{cl}(s)=\frac{\frac{\omega_c}{s}e^{-\tau s}}{1+\frac{\omega_c}{s}e^{-\tau s}}=\frac{\omega_c e^{-\tau s}}{s+\omega_c e^{-\tau s}}.
\end{aligned} 
\end{equation}
An approximation of (\ref{eq:e21xqx}) can be given as  
\begin{equation}\label{eq:e21xqqx}
\begin{aligned}
\hat{G}_{cl}(s)=\frac{b_ms^m+b_{m-1}s^{m-1}+...+b_0}{s^n+a_{n-1}s^{n-1}+...+a_0}e^{-\tau s},
\end{aligned} 
\end{equation}
where $ n=n_h+n_p $ and $ m\leq n $ are positive real constants, and $ a_{i} $ and $ b_j $ for $ i=0, ..., n-1 $ and $ j=0, ..., m-1 $, are real constants to be estimated. 
The reference model then can be obtained as the state space representation of (\ref{eq:e21xqqx}) as
\begin{equation}\label{eq:e7x}
\begin{aligned}
\dot{x}_{m}(t)&=A_{m}x_{m}(t)+B_{m}r(t-\tau),\\
\end{aligned} 
\end{equation}
where $ x_m\in \mathbb{R}^{(n_h+n_p)} $ is the reference model state vector, $ A_m\in \mathbb{R}^{(n_h+n_p)\times (n_h+n_p)} $ is the state matrix, $ B_m\in \mathbb{R}^{(n_h+n_p)\times m_h} $ is the input matrix, and $ r\in \mathbb{R}^{m_h} $ is the reference input.

\section{REFERENCE MODEL PARAMETERS}\label{sec:reference}
The crossover transfer function (\ref{eq:e1x}) contains the crossover frequency, $ \omega_c $, which is not known a priori. Experimental data, showing the reference input ($ r(t) $) frequency bandwidth, $ \omega_i $, versus crossover frequency $ \omega_c $, is provided in \cite{Bee08} and \cite{Mcr67}, for plant transfer functions $ K $, $ K/s $ and $ K/s^2 $. We fit polynomials to these experimental results to obtain the crossover frequency of the open loop transfer function given a reference input frequency bandwidth. These polynomials are given in Table I. It is noted that when the reference input has multiple frequency components, the highest frequency is used to calculate the crossover frequency.

\begin{remark}
	In this work, we use the polynomial relationships provided in Table I for zero, first and second order plant dynamics with nonzero poles and zeros. Further experimental work can be conducted to obtain a more precise relationship between the crossover and reference input frequencies, but this is currently out of the scope of this work.
\end{remark}

%
\begin{table}
	\vspace{0.3cm}
	\caption{} 
	\centering
	\begin{tabular}{ | c | c | }
		\hline 
		Plant transfer & Crossover frequency of the \\
		function & open loop transfer function (rad/s)\\
		\hline
		$ K $ & $ \omega_c=0.067\omega_i^2+0.099\omega_i+4.8 $ \\ \hline
		$ K/s $ & $ \omega_c=0.14\omega_i+4.3 $ \\ \hline
		$ K/s^2 $ & $ \omega_c=-0.0031\omega_i^4-0.072\omega_i^3+0.29\omega_i^2$ \\ & $-0.13\omega_i+3 $ \\
		\hline
	\end{tabular}
\end{table}

\section{HUMAN PILOT CONTROL DECISION COMMAND}\label{sec:human}

The adaptive human pilot control decision command, $ u(t) $, is determined as
\begin{equation}\label{eq:e8x}
u(t)=K_rK_xx_{hp}(t+\tau)+K_rr(t)
\end{equation}
where $ K_x\in \mathbb{R}^{1\times (n_h+n_p)} $, and $ K_r\in \mathbb{R}^{m_h\times m_h} $. Using (\ref{eq:e8x}) and (\ref{eq:e5x}), the closed loop dynamics can be obtained as
\begin{equation}\label{eq:e9x}
\begin{aligned}
\dot{x}_{hp}(t)&=(A_{hp}+B_{hp}K_rK_x)x_{hp}(t)+B_{hp}K_rr(t-\tau).
\end{aligned} 
\end{equation} 

Equation (\ref{eq:e8x}) describes a non-causal decision command which requires future values of the states. This problem can be eliminated by solving the differential equation (\ref{eq:e5x}) as a $ \tau $-seconds ahead predictor as
\begin{equation}\label{eq:e9xx}
\begin{aligned}
x_{hp}(t+\tau)&=e^{A_{hp}\tau}x_{hp}(t)+\int_{-\tau}^{0}e^{-A_{hp}\eta}B_{hp}u(t+\eta)d\eta. 
\end{aligned} 
\end{equation}

\begin{assumption}
There exist ideal parameters $ K_r^* $ and $ K_x^* $ satisfying the following matching conditions
\begin{equation}\label{eq:e10x}
\begin{aligned}
&A_{hp}+B_{hp}K_r^*K_x^*=A_m\\
&B_{hp}K_r^*=B_m.
\end{aligned} 
\end{equation}
\end{assumption} 


By substituting (\ref{eq:e9xx}) into (\ref{eq:e8x}), the human pilot control decision input can be written as 
\begin{equation}\label{eq:e11xxxy}
\begin{aligned}
u(t)&=K_rK_xe^{A_{hp}\tau}x_{hp}(t)\\
&+K_rK_x\int_{-\tau}^{0}e^{-A_{hp}\eta}B_{hp}u(t+\eta)d\eta+K_rr(t).
\end{aligned} 
\end{equation}
By defining $ \theta_x(t) $ and $ \lambda(t,\eta) $ as
\begin{equation}\label{eq:e11xxxq}
\begin{aligned}
\theta_x(t)&=K_r(t)K_x(t)e^{A_{hp}\tau},\\
\lambda(t,\eta)&=K_r(t)K_x(t)e^{-A_{hp}\eta}B_{hp},
\end{aligned} 
\end{equation}
(\ref{eq:e11xxxy}) can be rewritten as (see fig. \ref{fig:f1-2})
\begin{equation}\label{eq:e11xxx}
\begin{aligned}
u(t)=\theta_x(t)x_{hp}(t)+\int_{-\tau}^{0}\lambda(t,\eta)u(t+\eta)d\eta+K_r(t)r(t).
\end{aligned} 
\end{equation}
The ideal values of $ \theta_x $ and $ \lambda $ can be obtained as
\begin{equation}\label{eq:e11xxxxx}
\begin{aligned}
\theta_x^*&=K_r^*K_x^*e^{A_{hp}\tau}\\ 
\lambda^*(\eta)&=K_r^*K_x^*e^{-A_{hp}\eta}B_{hp}.
\end{aligned} 
\end{equation}
Since $ A_{hp} $ and $ B_{hp} $ are unknown, $ \theta_x $ and $ \lambda $ need to be estimated. The closed loop dynamics can be obtained using (\ref{eq:e5x}) and (\ref{eq:e11xxx}) as
\begin{equation}\label{eq:e11xxxx}
\begin{aligned}
\dot{x}_{hp}(t)&=A_{hp}x_{hp}(t)+B_{hp}\theta_x(t-\tau)x_{hp}(t-\tau)\\ &+\int_{-\tau}^{0}B_{hp}\lambda(t-\tau,\eta)u(t+\eta-\tau)d\eta\\ &+B_{hp}K_rr(t-\tau),
\end{aligned} 
\end{equation}

Defining the deviations of the adaptive parameters from their ideal values as $ \tilde{\theta}_x=\theta_x-\theta_x^* $ and $ \tilde{\lambda}=\lambda-\lambda^* $, and adding and subtracting $ A_mx_{hp}(t) $ to (\ref{eq:e11xxxx}), and using (\ref{eq:e10x}), we obtain that
\begin{equation}\label{eq:e12xxz}
\begin{aligned}
\dot{x}_{hp}(t)&=A_mx_{hp}(t)-B_{hp}K_r^*K_x^*x_{hp}(t)\\ &+B_{hp}K_r(t-\tau)K_x(t-\tau)\Big(e^{A_{hp}\tau}x_{hp}(t-\tau)\\ &+\int_{-\tau}^{0}e^{-A_{hp}\eta}B_{hp}u(t+\eta-\tau)d\eta \Big)\\ &+B_{hp}K_r(t-\tau)r(t-\tau).\\
\end{aligned} 
\end{equation}
 Using (\ref{eq:e9xx}), (\ref{eq:e12xxz}) is rewritten as
\begin{equation}\label{eq:e12xx}
\begin{aligned}
\dot{x}_{hp}(t)&=A_mx_{hp}(t)-B_{hp}K_r^*K_x^*x_{hp}(t)\\ &+B_{hp}K_r(t-\tau)K_x(t-\tau)x_{hp}(t)\\ &+B_{hp}K_r(t-\tau)r(t-\tau).
\end{aligned} 
\end{equation}
Defining the tracking error as $ e(t)=x_{hp}-x_{m} $, and subtracting (\ref{eq:e7x}) from (\ref{eq:e12xx}), and using (\ref{eq:e10x}), and following a similar procedure given in \cite{NarAnn12}, it is obtained that
\begin{equation}\label{eq:e12xxxz}
\begin{aligned}
\dot{e}(t)&=\dot{x}_{hp}-\dot{x}_{m}\\
&=A_me(t)-B_{hp}K_r^*K_x^*x_{hp}(t)\\ &+B_{hp}K_r(t-\tau)K_x(t-\tau)x_{hp}(t)\\ &+B_{hp}(K_r(t-\tau)-K_r^*)r(t-\tau)\\
&=A_me(t)+\big( -B_{hp}K_r^*K_x^*\\
&+B_{hp}(K_r^*-K_r^*+K_r(t-\tau))K_x(t-\tau)\big)x_{hp}(t) 
\\&+B_{hp}(K_r(t-\tau)-K_r^*)r(t-\tau)\\
&=A_me(t)+B_{m}(K_x(t-\tau)-K_x^*)x_{hp}(t)\\
&+B_{m}({K_r^*}^{-1}K_r(t-\tau)-1)K_x(t-\tau)x_{hp}(t)\\
&+B_{m}({K_r^*}^{-1}K_r(t-\tau)-1)r(t-\tau)\\
&=A_me(t)+B_{m}(\tilde{K}_x(t-\tau)x_{hp}(t)\\
&\hspace{-0.1cm}+B_{m}({K_r^*}^{-1}-K_r^{-1}(t-\tau))K_r(t-\tau)K_x(t-\tau)x_{hp}(t)\\
&+B_{m}({K_r^*}^{-1}-K_r^{-1}(t-\tau))K_r(t-\tau)r(t-\tau).
\end{aligned} 
\end{equation}
Using (\ref{eq:e9xx}) and defining $ \Phi={K_r^*}^{-1}-K_r^{-1} $, we can rewrite (\ref{eq:e12xxxz}) as
\begin{equation}\label{eq:e12xxxzz}
\begin{aligned}
\dot{e}(t)&=A_me(t)+{B_mK_r^*}^{-1}(K_r^*K_x(t-\tau)-K_r^*K_x^*)\\
&\hspace{-0.1cm}\times\Big( e^{A_{hp}\tau}x_{hp}(t-\tau)+\int_{-\tau}^{0}e^{-A_{hp}\eta}B_{hp}u(t+\eta-\tau)d\eta \Big)\\
&+B_m\Phi(t-\tau)\Big( K_r(t-\tau)K_x(t-\tau)\Big( e^{A_{hp}\tau}x_{hp}(t-\tau)\\
&+\int_{-\tau}^{0}e^{-A_{hp}\eta}B_{hp}u(t+\eta-\tau)d\eta \Big)\\
&+K_r(t-\tau)r(t-\tau) \Big).
\end{aligned} 
\end{equation}
Using (\ref{eq:e11xxxxx}) and (\ref{eq:e12xxxzz}), we obtain that
\begin{equation}\label{eq:e12xxxzzz}
\begin{aligned}
\dot{e}(t)&=A_me(t)+B_mK_x(t-\tau)\Big( e^{A_{hp}\tau}x_{hp}(t-\tau)\\
&+\int_{-\tau}^{0}e^{-A_{hp}\eta}B_{hp}u(t+\eta-\tau)d\eta \Big)\\
&-B_m{K_r^*}^{-1}\Big( \theta_x^*x_{hp}(t-\tau)\\
&+\int_{-\tau}^{0}\lambda^*(\eta)u(t+\eta-\tau)d\eta \Big)\\
&+B_{m}\Phi(t-\tau)u(t-\tau).\\
\end{aligned} 
\end{equation}
Using (\ref{eq:e11xxxq}), (\ref{eq:e12xxxzzz}) can be rewritten as
\begin{equation}\label{eq:e12xxxq}
\begin{aligned}
\dot{e}(t)&=A_me(t)+B_m\Big( \big( K_r^{-1}(t-\tau)\theta_x(t-\tau)-{K_r^*}^{-1}\theta_x^*\big)\\
&\times x_{hp}(t-\tau)+\int_{-\tau}^{0}\big( K_r^{-1}(t-\tau)\lambda(t-\tau,\eta)\\
&-{K_r^*}^{-1}\lambda^*(\eta) \big) u(t+\eta-\tau)d\eta \Big)\\
&+B_{m}\Phi(t-\tau)u(t-\tau).
\end{aligned} 
\end{equation}
Defining $ \theta_1=K_r^{-1}\theta_x $ and $ \lambda_1=K_r^{-1}\lambda $, and using their deviations from their ideal values, $ \tilde{\theta}_1=\theta_1-\theta_1^* $ and $ \tilde{\lambda}_1=\lambda_1-\lambda_1^* $, where 
$ \theta_1^*={K_r^*}^{-1}\theta_x^* $ and $ \lambda_1^*={K_r^*}^{-1}\lambda^* $, (\ref{eq:e12xxxq}) can be rewritten as
\begin{equation}\label{eq:e12xxx}
\begin{aligned}
\dot{e}(t)&=A_me(t)+B_{m}\tilde{\theta}_1(t-\tau)x_{hp}(t-\tau)\\ &+B_m\int_{-\tau}^{0}\tilde{\lambda}_1(t-\tau,\eta)u(t+\eta-\tau)d\eta \\ 
&+B_{m}\Phi(t-\tau)u(t-\tau).
\end{aligned} 
\end{equation}
The following lemma will be necessary to prove the main theorem of this article.

\begin{lemma}\label{lem1}
Suppose that the continuous function $ u(t) $ is given as
\begin{equation}\label{eq:e21xx}
\begin{aligned}
u(t)=f(t)+\int_{-\tau}^{0}\lambda(t,\eta)u(t+\eta)d\eta,
\end{aligned} 
\end{equation}
where $ u,f: [t_0-\tau,\infty]\rightarrow R $, and $ \lambda:[t_0, \infty)\times [-\tau, 0]\rightarrow R $. Then
\begin{equation}\label{eq:e22xx}
\begin{aligned}
|u(t)|\leq 2(\bar{f}+c_0c_1)e^{c_0^2(t-t')},\ \forall t_j'\geq t_i',
\end{aligned} 
\end{equation}
if constants $ t_i',\bar{f}, c_0, c_1\in R^+ $ exist such that $ |f(t)|\leq \bar{f} $, 
\begin{equation}\label{eq:e23xx}
\begin{aligned}
\int_{-\tau}^{0}\lambda^2(t,\eta)d\eta\leq c_0^2 \ \ for\  t\in [t_i',t_j'),
\end{aligned} 
\end{equation}
and
\begin{equation}\label{eq:e24xx}
\begin{aligned}
\int_{-\tau}^{0}u^2(t+\eta)d\eta\leq c_1^2 \ \ \forall
t\leq t_i'.
\end{aligned} 
\end{equation}
\end{lemma}
\begin{proof}
The proof of Lemma \ref{lem1} can be found in \cite{YilAnn10}.
%
\end{proof}


\begin{theorem}\label{thm1}
Given the initial conditions $ \tilde{\theta}_1(\xi)$, $ \tilde{\lambda}_1(\xi,\eta) $, $ \Phi(\xi) $ and $ x_{hp}(\xi) $ for $ \xi\in [-\tau, 0] $, and $ u(\zeta) $ for $ \zeta\in[-2\tau,0] $, there exists a $ \tau^* $ such that for all $ \tau\in[0,\tau^*] $, the controller (\ref{eq:e11xxx}) with the adaptive laws
\begin{equation}\label{eq:14thm}
\dot{{\theta}}_1^T(t)=-x_{hp}(t-\tau)e(t)^TPB_{m},
\end{equation} 
\begin{equation}\label{eq:15thm}
\dot{{\Phi}}^T(t)=-u(t-\tau)e(t)^TPB_{m},\\
\end{equation} 
\begin{equation}\label{eq:16thm}
\dot{\lambda}_1^T(t,\eta)=-u(t+\eta-\tau)e(t)^TPB_{m},
\end{equation} 
where $ P $ is the symmetric positive definite matrix satisfying
the Lyapunov equation $ A_m^TP +P A_m = -Q $ for a symmetric
positive definite matrix $ Q $, which can be employed to obtain controller parameters using $ \dot{K}_r=\text{Proj}(K_r\dot{\Phi}K_r) $, $\theta_x(t)=K_r(t)\theta_1(t) $ and $ \lambda(t)=K_r(t)\lambda_1(t) $, make the pilot neuromuscular and plant aggregate system (\ref{eq:e5x}) follow the crossover reference model (\ref{eq:e7x}) asymptotically, i.e, $ lim_{t\to \infty} x_{hp}(t)=x_m(t) $, while keeping all the signals bounded.
\end{theorem}
\begin{proof}
Consider a Lyapunov-Krasovskii functional \cite{YilAnn10}
\begin{equation}\label{eq:e12x}
\begin{aligned}
V(t)&=e^TPe+\text{tr}({\Phi}^T(t){\Phi}(t))+\text{tr}(\tilde{\theta}_1^T(t)\tilde{\theta}_1(t))\\ &+\int_{-\tau}^{0}\int_{t+v}^{t}\text{tr}(\dot{\tilde{\theta}}_1^T(\xi)\dot{\tilde{\theta}}_1(\xi))d\xi dv\\
&+\int_{-\tau}^{0}\int_{t+v}^{t}\text{tr}(\dot{{\Phi}}^T(\xi)\dot{{\Phi}}(\xi))d\xi dv\\
&+\int_{-\tau}^{0}\text{tr}(\tilde{\lambda}_1^T(t,\eta)\tilde{\lambda}_1(t,\eta))d\eta\\ &+\int_{-\tau}^{0}\int_{t+v}^{t}\int_{-\tau}^{0}\text{tr}(\dot{\tilde{\lambda}}_1^T(\xi,\eta)\dot{\tilde{\lambda}}_1(\xi,\eta))d\eta d\xi dv.
\end{aligned} 
\end{equation}
The derivative of $ V(t) $ can be calculated as
\begin{equation}\label{eq:e14x}
\begin{aligned}
\dot{V}(t)&=\dot{e}^T(t)^TPe(t)+e^T(t)P\dot{e}(t)+2\text{tr}(\dot{\tilde{\theta}}_1^T(t)\tilde{\theta}_1(t))\\ &+2\text{tr}(\dot{{\Phi}}^T(t){\Phi}(t))+\int_{-\tau}^{0}2\text{tr}(\dot{\tilde{\lambda}}_1^T(t,\eta)\tilde{\lambda}_1(t,\eta))d\eta\\ 
&+\tau \text{tr}(\dot{\tilde{\theta}}_1^T(t)\dot{\tilde{\theta}}_1(t))-\int_{-\tau}^{0}\text{tr}(\dot{\tilde{\theta}}_1^T(t+v)\dot{\tilde{\theta}}_1(t+v))dv\\ 
&+\tau\text{tr}(\dot{{\Phi}}^T(t)\dot{{\Phi}}(t))-\int_{-\tau}^{0}\text{tr}(\dot{{\Phi}}^T(t+v)\dot{{\Phi}}(t+v))dv\\
&+\tau\int_{-\tau}^{0}\text{tr}(\dot{\tilde{\lambda}}_1^T(t,\eta)\dot{\tilde{\lambda}}_1(t,\eta))d\eta\\
&-\int_{-\tau}^{0}\int_{-\tau}^{0}\text{tr}(\dot{\tilde{\lambda}}_1^T(t+v,\eta)\dot{\tilde{\lambda}}_1(t+v,\eta))d\eta dv.
\end{aligned} 
\end{equation}
Substituting (\ref{eq:e12xxx}) into (\ref{eq:e14x}) and using the Lyapunov equation $ A_m^TP +P A_m = -Q $, it is obtained that
\begin{equation}
\begin{aligned}
\dot{V}(t)&=-{e}^T(t)Qe(t)+2e^T(t)PB_m\tilde{\theta}_1(t-\tau)x_{hp}(t-\tau)\\
&+2e^T(t)PB_m\int_{-\tau}^{0}\tilde{\lambda}_1(t-\tau,\eta)u(t+\eta-\tau)d\eta\\
&+2e^T(t)PB_m\Phi(t-\tau)u(t-\tau)\\
&+2\text{tr}(\dot{\tilde{\theta}}_1^T(t)\tilde{\theta}_1(t))+2\text{tr}(\dot{{\Phi}}^T(t){\Phi}(t)) \\
&+\int_{-\tau}^{0}2\text{tr}(\dot{\tilde{\lambda}}_1^T(t,\eta)\tilde{\lambda}_1(t,\eta))d\eta\\ 
&+\tau \text{tr}(\dot{\tilde{\theta}}_1^T(t)\dot{\tilde{\theta}}_1(t))-\int_{-\tau}^{0}\text{tr}(\dot{\tilde{\theta}}_1^T(t+v)\dot{\tilde{\theta}}_1(t+v))dv\\ 
&+\tau\text{tr}(\dot{{\Phi}}^T(t)\dot{{\Phi}}(t))-\int_{-\tau}^{0}\text{tr}(\dot{{\Phi}}^T(t+v)\dot{{\Phi}}(t+v))dv\notag
\end{aligned}
\end{equation}
\begin{equation}\label{eq:e14xx}
\begin{aligned}
&+\tau\int_{-\tau}^{0}\text{tr}(\dot{\tilde{\lambda}}_1^T(t,\eta)\dot{\tilde{\lambda}}_1(t,\eta))d\eta\\
&-\int_{-\tau}^{0}\int_{-\tau}^{0}\text{tr}(\dot{\tilde{\lambda}}_1^T(t+v,\eta)\dot{\tilde{\lambda}}_1(t+v,\eta))d\eta dv.
\end{aligned} 
\end{equation}
Using $ g(t)-g(t-\tau)=\int_{-\tau}^{0}\dot{g}(t+v)dv $, (\ref{eq:e14xx}) can be rewritten as
\begin{equation}\label{eq:e14xxxx}
\begin{aligned}
\dot{V}(t)&=-{e}^T(t)Qe(t)\\
&+2\text{tr}\Big(x_{hp}(t-\tau)e^T(t)PB_m\tilde{\theta}_1(t)+\dot{\tilde{\theta}}_1^T(t)\tilde{\theta}_1(t)\Big)\\
&+2\text{tr}\Big(u(t-\tau)e^T(t)PB_m\Phi(t)+\dot{{\Phi}}^T(t){\Phi}(t)\Big)\\
&+\int_{-\tau}^{0}2\text{tr}\Big(u(t+\eta-\tau)e^T(t)PB_m\tilde{\lambda}_1(t,\eta)\\ &+\dot{\tilde{\lambda}}_1^T(t,\eta)\tilde{\lambda}_1(t,\eta)\Big)d\eta\\
&-2e^T(t)PB_m(\int_{-\tau}^{0}\dot{\tilde{\theta}}_1(t+v)dv)x_{hp}(t-\tau)\\
&-2e^T(t)PB_m(\int_{-\tau}^{0}\dot{{\Phi}}(t+v)dv)u(t-\tau)\\
&-2e^T(t)PB_m\Big(\int_{-\tau}^{0}(\int_{-\tau}^{0}\dot{\tilde{\lambda}}_1(t+v,\eta)dv)\\
&\times u(t+\eta-\tau)d\eta\Big)\\
&+\tau \text{tr}(\dot{\tilde{\theta}}_1^T(t)\dot{\tilde{\theta}}_1(t))-\int_{-\tau}^{0}\text{tr}(\dot{\tilde{\theta}}_1^T(t+v)\dot{\tilde{\theta}}_1(t+v))dv\\ 
&+\tau\text{tr}(\dot{{\Phi}}^T(t)\dot{{\Phi}}(t))-\int_{-\tau}^{0}\text{tr}(\dot{{\Phi}}^T(t+v)\dot{{\Phi}}(t+v))dv\\
&+\tau\int_{-\tau}^{0}\text{tr}(\dot{\tilde{\lambda}}_1^T(t,\eta)\dot{\tilde{\lambda}}_1(t,\eta))d\eta\\
&-\int_{-\tau}^{0}\int_{-\tau}^{0}\text{tr}(\dot{\tilde{\lambda}}_1^T(t+v,\eta)\dot{\tilde{\lambda}}_1(t+v,\eta))d\eta dv.
\end{aligned} 
\end{equation}
By substituting (\ref{eq:14thm})-(\ref{eq:16thm}) into (\ref{eq:e14xxxx}), it is obtained that
\begin{equation}
\begin{aligned}
\dot{V}(t)&=-{e}^T(t)Qe(t)\\
&-2\int_{-\tau}^{0}\text{tr}(x_{hp}(t-\tau)e(t)^TPB_{m}\dot{\tilde{\theta}}_1(t+v))dv\\ &-2\int_{-\tau}^{0}\text{tr}(u(t-\tau)e(t)^TPB_{m}\dot{{\Phi}}(t+v))dv\\
&-2\int_{-\tau}^{0}\int_{-\tau}^{0}\text{tr}(u(t+\eta-\tau)e(t)^TPB_{m}\dot{\tilde{\lambda}}_1(t+v,\eta))dvd\eta\\
&+\tau \text{tr}(\dot{\tilde{\theta}}_1^T(t)\dot{\tilde{\theta}}_1(t))-\int_{-\tau}^{0}\text{tr}(\dot{\tilde{\theta}}_1^T(t+v)\dot{\tilde{\theta}}_1(t+v))dv\\ 
&+\tau\text{tr}(\dot{{\Phi}}^T(t)\dot{{\Phi}}(t))-\int_{-\tau}^{0}\text{tr}(\dot{{\Phi}}^T(t+v)\dot{{\Phi}}(t+v))dv\\
&+\tau\int_{-\tau}^{0}\text{tr}(\dot{\tilde{\lambda}}_1^T(t,\eta)\dot{\tilde{\lambda}}_1(t,\eta))d\eta\\
&-\int_{-\tau}^{0}\int_{-\tau}^{0}\text{tr}(\dot{\tilde{\lambda}}_1^T(t+v,\eta)\dot{\tilde{\lambda}}_1(t+v,\eta))d\eta dv\notag 
\end{aligned} 
\end{equation}
\begin{equation}\label{eq:e17x}
\begin{aligned}
&=-{e}^T(t)Qe(t)+\int_{-\tau}^{0}\text{tr}\Big(2\dot{\tilde{\theta}}_1^T(t)\dot{\tilde{\theta}}_1(t+v)\\
&+\dot{\tilde{\theta}}_1^T(t)\dot{\tilde{\theta}}_1(t) -\dot{\tilde{\theta}}_1^T(t+v)\dot{\tilde{\theta}}_1(t+v)\Big)dv\\
&+\int_{-\tau}^{0}\text{tr}\Big(2\dot{{\Phi}}^T(t)\dot{{\Phi}}(t+v)+\dot{{\Phi}}^T(t)\dot{{\Phi}}(t)\\
&-\dot{{\Phi}}^T(t+v)\dot{{\Phi}}(t+v)\Big)dv\\
&+\int_{-\tau}^{0}\int_{-\tau}^{0}\text{tr}\Big(2\dot{\tilde{\lambda}}_1^T(t,\eta)\dot{\tilde{\lambda}}_1(t+v,\eta)\\
&+\dot{\tilde{\lambda}}_1^T(t,\eta)\dot{\tilde{\lambda}}_1(t,\eta)-\dot{\tilde{\lambda}}_1^T(t+v,\eta)\dot{\tilde{\lambda}}_1(t+v,\eta)\Big)d\eta dv.
\end{aligned} 
\end{equation}
Using the trace property $ \text{tr}(A+B)=\text{tr}(A)+\text{tr}(B) $, and the algebraic inequality $ a^2\geq 2ab-b^2 $ for two scalars $ a $ and $ b $, it can be shown that $ \text{tr}(2A^TB+A^TA-B^TB)\leq 2\text{tr}(A^TA) $. Using these inequalities, (\ref{eq:e17x}) can be rewritten as
\begin{equation}\label{eq:e17xx}
\begin{aligned}
\dot{V}(t)&\leq-{e}^T(t)Qe(t)+\int_{-\tau}^{0}2\text{tr}(\dot{\tilde{\theta}}_1^T(t)\dot{\tilde{\theta}}_1(t))dv\\
&+\int_{-\tau}^{0}2\text{tr}(\dot{{\Phi}}^T(t)\dot{{\Phi}}(t))dv\\
&+\int_{-\tau}^{0}\int_{-\tau}^{0}2\text{tr}(\dot{\tilde{\lambda}}_1^T(t,\eta)\dot{\tilde{\lambda}}_1(t,\eta))d\eta dv.
\end{aligned} 
\end{equation}
By substituting (\ref{eq:14thm})-(\ref{eq:16thm}) into (\ref{eq:e17xx}), and using the trace operator property $ \text{tr}(AB)=\text{tr}(BA) $ for two square matrices $ A $ and $ B $, (\ref{eq:e17xx}) can be rewritten as
\begin{equation}\label{eq:e18x}
\begin{aligned}
\dot{V}(t)&\leq-{e}^T(t)Qe(t)\\
&+2\tau \text{tr}\big(e(t)x_{hp}^T(t-\tau)x_{hp}(t-\tau)e(t)^TPB_{m}B_{m}^TP\big)\\ 
&+2\tau \text{tr}\big(e(t)u^T(t-\tau)u(t-\tau)e(t)^TPB_{m}B_{m}^TP\big)\\
&+2\tau\int_{-\tau}^{0} \text{tr}\big(e(t)u^T(t+\eta-\tau)u(t+\eta-\tau)e(t)^T\\
&\times PB_{m}B_{m}^TP\big)d\eta.
\end{aligned} 
\end{equation}
Using $ \text{tr}(AB)\leq \text{tr}(A)\text{tr}(B) $ for two positive semidefinite matrices $ A $ and $ B $, and $ \text{tr}(X^TX)=||X||_F^2 $ for a matrix $ X $, an upper bound for (\ref{eq:e18x}) can be derived as 
\begin{equation}
\begin{aligned}
\dot{V}(t)&\leq-{e}^T(t)Qe(t)\\
&+2\tau \text{tr}(e(t)x_{hp}^T(t-\tau)x_{hp}(t-\tau)e(t)^T)\text{tr}(PB_{m}B_{m}^TP)\\ 
&+2\tau \text{tr}\big(e(t)u^T(t-\tau)u(t-\tau)e(t)^T\big)\text{tr}\big(PB_{m}B_{m}^TP\big)\\
&+2\tau\int_{-\tau}^{0} \text{tr}\big(e(t) u^T(t-\tau+\eta)u(t-\tau+\eta)e(t)^T\big)\\
&\times \text{tr}\big(PB_{m}B_{m}^TP\big)d\eta\\
&\leq -\lambda_{min}(Q)||e(t)||^2\\
&+2\tau||x_{hp}(t-\tau)e(t)^T||_F^2||B_{m}^TP||_F^2\\
&+2\tau ||u(t-\tau)e(t)^T||_F^2||B_{m}^TP||_F^2\\
&+2\tau\int_{-\tau}^{0}||u(t+\eta-\tau)e(t)^T||_F^2||B_{m}^TP||_F^2d\eta \notag
\end{aligned}
\end{equation}
\begin{equation}\label{eq:e19x}
\begin{aligned}
&\leq -\lambda_{min}(Q)||e(t)||^2\\
&+2\tau||x_{hp}(t-\tau)||^2||e(t)||^2||B_{m}^TP||_F^2\\
&+2\tau ||u(t-\tau)||^2||e(t)||^2||B_{m}^TP||_F^2\\
&+2\tau\int_{-\tau}^{0}||u(t+\eta-\tau)||^2||e(t)||^2||B_{m}^TP||_F^2d\eta\\
&=||B_{m}^TP||_F^2||e(t)||^2\Big( -\frac{\lambda_{min}(Q)}{||B_{m}^TP||_F^2}\\
&+2\tau\big( ||x_{hp}(t-\tau)||^2+||u(t-\tau)||^2\\
&+\int_{-\tau}^{0}||u(t+\eta-\tau)||^2d\eta \big) \Big).
\end{aligned} 
\end{equation}
Defining $ q\equiv \frac{\lambda_{min}(Q)}{||B_{m}^TP||_F^2} $, the inequality
\begin{equation}\label{eq:e20x}
\begin{aligned}
q&-2\tau\big( ||x_{hp}(t-\tau)||^2+||u(t-\tau)||^2+\\
&+\int_{-\tau}^{0}||u(t+\eta-\tau)||^2d\eta \big)> 0.
\end{aligned} 
\end{equation}
needs to be satisfied for the non-positiveness of $ \dot{V} $.
Assuming that $ x_{hp} $ and $ u $ are bounded in the interval $ [t_0-2\tau,t_0) $, the rest of the proof is divided into the following four steps:

\noindent
\textbf{Step 1} In this step, the negative semi-definiteness of the Lyapunov-Krasovskii functional's (\ref{eq:e12x}) time derivative in the interval $ [t_0-\tau,t_0) $ is shown which leads to the boundedness of the the signals in this interval. In addition, an upper bound for $ u $ in the interval $ [t_0-2\tau,t_0) $ is given.

Suppose that
\begin{equation}
\begin{aligned}
\sup_{\xi\in[t_0-\tau, t_0)}&||x_{hp}(\xi)||^2\leq \gamma_1\\
\sup_{\xi\in[t_0-2\tau, t_0)}&||u(\xi)||^2\leq \gamma_2
\end{aligned} 
\end{equation} 
for some positive $ \gamma_1, \gamma_2$, and a $ \tau_1 $ is given such that
\begin{equation}
\begin{aligned}
2\tau_1(\gamma_1+\gamma_2+\tau_1\gamma_2)<q.
\end{aligned} 
\end{equation} 
Then the following inequality is satisfied:
\begin{equation}
\begin{aligned}
q&-2\tau\big( ||x_{hp}(\xi-\tau)||^2+||u(\xi-\tau)||^2+\\
&+\int_{-\tau}^{0}||u(\xi+\eta-\tau)||^2d\eta \big)> 0,\\ &\forall \xi\in[t_0,t_0+\tau), \forall\tau\in[0, \tau_1].
\end{aligned} 
\end{equation}
It follows that $ V(t) $, defined in (\ref{eq:e12x}), is non-increasing for $ t\in[t_0, t_0+\tau) $. Thus, we have
\begin{equation}
\begin{aligned}
\lambda_{min}(P)||e(\xi)||^2\leq e(\xi)^TPe(\xi)\leq V(t_0),
\end{aligned} 
\end{equation}
which leads to
\begin{equation}
\begin{aligned}
||x_{hp}(\xi)||-||x_m(\xi)||\leq ||e(\xi)||\leq \sqrt{\frac{V(t_0)}{\lambda_{min}(P)}}.
\end{aligned} 
\end{equation}
Then, we have
\begin{equation}
\begin{aligned}
||x_{hp}(\xi)||\leq\sqrt{\frac{V(t_0)}{\lambda_{min}(P)}}+||x_{m}(\xi)||,
\end{aligned} 
\end{equation}
for $ \forall \xi\in[t_0,t_0+\tau) $.
We also have the inequality
\begin{align}\label{eq:x47}
&||{\Phi}(\xi)||^2\leq V(t_0)\implies ||{K_r^*}^{-1}-K_r^{-1}(\xi)||^2\leq V(t_0)\notag\\
&\implies ||K_r^{-1}(\xi)||\leq \sqrt{V(t_0)}+||{K_r^*}^{-1}||.
\end{align}
for $ \forall \xi\in[t_0,t_0+\tau) $.
It is noted that the boundedness of $ \Phi={K_r^*}^{-1}-K_r^{-1} $ does not guarantee the boundedness of $ \tilde{K}_r $. In order to guarantee the boundedness of $ \tilde{K}_r $ independent of the boundedness of $ \Phi $, the projection algorithm \cite{Eug13} is employed as 
\begin{equation}\label{eq:e19xzzz}
\begin{aligned}
\dot{K}_r=\text{Proj}\big(K_r,-K_rB_m^TPe(t)u^T(t-\tau)K_r\big),
\end{aligned} 
\end{equation}
with an upper bound $ K_{max} $, that is $ ||{K}_r||\leq K_{max} $. 
Thus, a lower bound for $ ||K_r^{-1}(\xi)|| $ can be calculated using the following algebraic manipulations: 
 \begin{align}\label{eq:47x}
 &K_r(\xi)K_r^{-1}(\xi)=I\Rightarrow ||K_r(\xi)K_r^{-1}(\xi)||=1\notag \\
 &\Rightarrow 1\leq ||K_r(\xi)||||K_r^{-1}(\xi)||\leq K_{max}||K_r^{-1}(\xi)||\notag \\
 & \Rightarrow \frac{1}{K_{max}}\leq ||K_r^{-1}(\xi)||.
 \end{align}
Defining $ k_1=\sqrt{V(t_0)}+||{K_r^*}^{-1}|| $, and using (\ref{eq:x47}), it is obtained that
\begin{align}\label{eq:47}
\frac{1}{K_{max}}\leq ||K_r^{-1}(\xi)||\leq k_1,\ \ \ \xi\in[t_0,t_0+\tau).
\end{align}
Therefore, $ K_r $ is always bounded and $ K_r^{-1}(\xi) $ is bounded for $ \forall \xi\in[t_0,t_0+\tau) $.

Furthermore, using the definitions of $ \theta_x, \theta_1, \lambda,\lambda_1 $ given in Theorem \ref{thm1}, and the non-increasing Lyapunov functional (\ref{eq:e12x}), it can be concluded that
\begin{align}\label{eq:48}
||\tilde{\theta}_1(\xi)||_F^2\leq V(t_0)&\implies ||\tilde{K}_r^{-1}(\xi)\tilde{\theta}_x(\xi)||_F^2\leq V(t_0),
\end{align}
\begin{align}\label{eq:49}
\int_{-\tau}^{0}||\tilde{\lambda}_1(\xi,\eta)||_F^2 &d\eta \leq V(t_0)\\
&\implies \int_{-\tau}^{0}||K_r^{-1}(\xi)\tilde{\lambda}(\xi,\eta)||_F^2d\eta\leq V(t_0),\notag
\end{align} 
for $ \forall \xi\in[t_0,t_0+\tau) $. Using (\ref{eq:48}) and (\ref{eq:49}), it can be obtained that
\begin{equation}
\begin{aligned}
||\tilde{\theta}_x(\xi)||_F^2&\leq K_{max}^2V(t_0),\\
\int_{-\tau}^{0}||\tilde{\lambda}(\xi,\eta)||_F^2d\eta&\leq K_{max}^2V(t_0).
\end{aligned} 
\end{equation}
for $ \forall \xi\in[t_0,t_0+\tau) $.

To simplify the notation, we define
\begin{equation}
\begin{aligned}
I_0\equiv \text{max}&\Big( \sqrt{\frac{V(t_0)}{\lambda_{min}(P)}}+\sup_{[t_0,t_0+\tau)}||x_{m}(\xi)||\\ &, K_{max}\sqrt{V(t_0)}, K_{max}^2V(t_0) \Big),
\end{aligned} 
\end{equation}
where $ R_{max} $ is the upper bound of the reference input $ r(t) $.

An upper bound on the control signal $ u(t) $ for $ t\in[t_0, t_0+\tau) $ can be derived by using Lemma \ref{lem1} and (\ref{eq:e11xxx}). In particular, setting $ t_i'=t_0 $, $ t_j'=t_0+\tau $, $ c_0^2=V(t_0) $, we obtain that
\begin{equation}\label{eq:e50}
\begin{aligned}
|u(\xi)|\leq 2\Big( \bar{f}+\big( \int_{-\tau}^{0}u^2(t_0+\eta)d\eta \big)^{1/2}I_0 \Big)e^{I_0\tau},
\end{aligned} 
\end{equation}
for $ \forall \xi\in[t_0,t_0+\tau) $, where $ \bar{f} $, which is the upper bound of $ \theta_x(t)x_{hp}(t)+K_r(t)r(t) $, depends only on $ I_0 $. Defining $ g(\gamma_2, I_0, \tau)\equiv 2(\bar{f}+\gamma_2I_0\sqrt{\tau})e^{I_0\tau} $, (\ref{eq:e50}) can be rewritten as
\begin{equation}
\begin{aligned}
|u(\xi)|\leq g(\gamma_2, I_0, \tau),\ \forall \xi\in[t_0,t_0+\tau).
\end{aligned} 
\end{equation}

The rest of the proof is similar to the one given in \cite{YilAnn10}. Below, a summary of the next steps are given.

\noindent
\textbf{Step 2} A delay range $ [0, \tau_2] $ is found that satisfies the condition (\ref{eq:e20x}) over the interval $ [t_0, t_0+2\tau) $ as
\begin{align}\label{eq:step2}
	2\tau_2\left( I_0^2+\left( max\left( \gamma_2,g\left( \gamma_2,I_0,\tau_2 \right)  \right)  \right)^2 (1+\tau_2) \right) < q,
\end{align} 
which leads to $ ||x_{hp}(\xi)||<I_0 $, $ \forall\xi\in[t_0, t_0+2\tau) $, $ \forall \tau\in [0, \bar{\tau}_2] $, $ \bar{\tau}_2=min\{\tau_1, \tau_2\} $.

\noindent
\textbf{Step 3} It is shown in this step that the bound on $ u $ over the interval $ [t_0, t_0+\tau) $ depends only on $ A_{hp} $, $ B_{hp} $, $ T $ and $ \tau $, where $ T $ is a value between $ t_0 $ and $ \tau $. Denoting this upper bound as $ U(I_0) $, we have $ |u(t)|\leq U(I_0) $, $ t\in[t_0,t_0+\tau) $.

\noindent
\textbf{Step 4} Using the calculated upper bound for $ u $ in the previous step, a delay range $ [0, \tau_2] $ is found that satisfies the condition 
\begin{align}\label{eq:step4}
2\tau_3\left( I_0^2+\left( max\left( U(I_0),g\left( U(I_0),I_0,\tau_3 \right)  \right)  \right)^2 (1+\tau_3) \right) < q.
\end{align} 
For $ \tau^*=min(\bar{\tau_2}, \tau_3) $, $ ||x_{hp}(\xi)||\leq I_0 $ and $ |u(\xi)|\leq U(I_0) $ for all $ \xi \in [t_0, t_0+\tau] $, $ \forall \tau\in [0, \tau^*] $.

%

The above four steps show that $ x_{hp}(\xi) $ and $ u(\xi) $ are bounded for $ \forall\xi\in[t_0, t_0+k\tau] $, for $ k=1 $ and $ \tau  \in (0, \tau^*]$. By assuming that $ x_{hp} $ and $ u $ are bounded for a given $ k $, the rest of the proof consists of showing that the boundedness of these variables hold for $ k+1 $. Using this assumption and repeating steps 1-4, which leads to satisfying (\ref{eq:step4}), we conclude that the Lyapunov function is non-increasing and $ ||x_{hp}(\xi)|| \leq I_0 $, and $ |u(\xi)|\leq g(U(I_0),I_0,\tau) $ for $ \xi\in [t_0, t_0+(k+1)\tau] $, $ \tau \leq \tau^* \leq \tau_3 $.
This completes the boundedness proof using induction.
Then, using Barbalat’s Lemma, it can be shown that the error between the human-in-the-loop system output $ x_{hp} $ and the reference model output $ x_m $ converges to zero.
\end{proof}

\section{Experimental Results}\label{sec:experiment}
\subsection{Experimental environment}

In order to test the proposed adaptive human model against data, an experimental setup consisting of a Logitech Extreme 3D Pro joystick and a Toshiba Portege-Z30-B laptop with Intel Core i7 CPU is used (see Fig. \ref{fig:setup}). 



The tracking task is performed by an operator monitoring the compensatory display, which provides information about the error between the target to be followed, and follower, which is the output of the plant (see Fig. \ref{fig:attitudeind}). The operator provides the input $ u_p $ (see Fig. \ref{fig:f1-2}) through the joystick, which is fed to the plant using MATLAB SIMULINK (R2018b). In return, the response of the plant is calculated and shown on the laptop screen in real-time.
\begin{figure}
	\begin{center}
		\includegraphics[width=8.5cm]{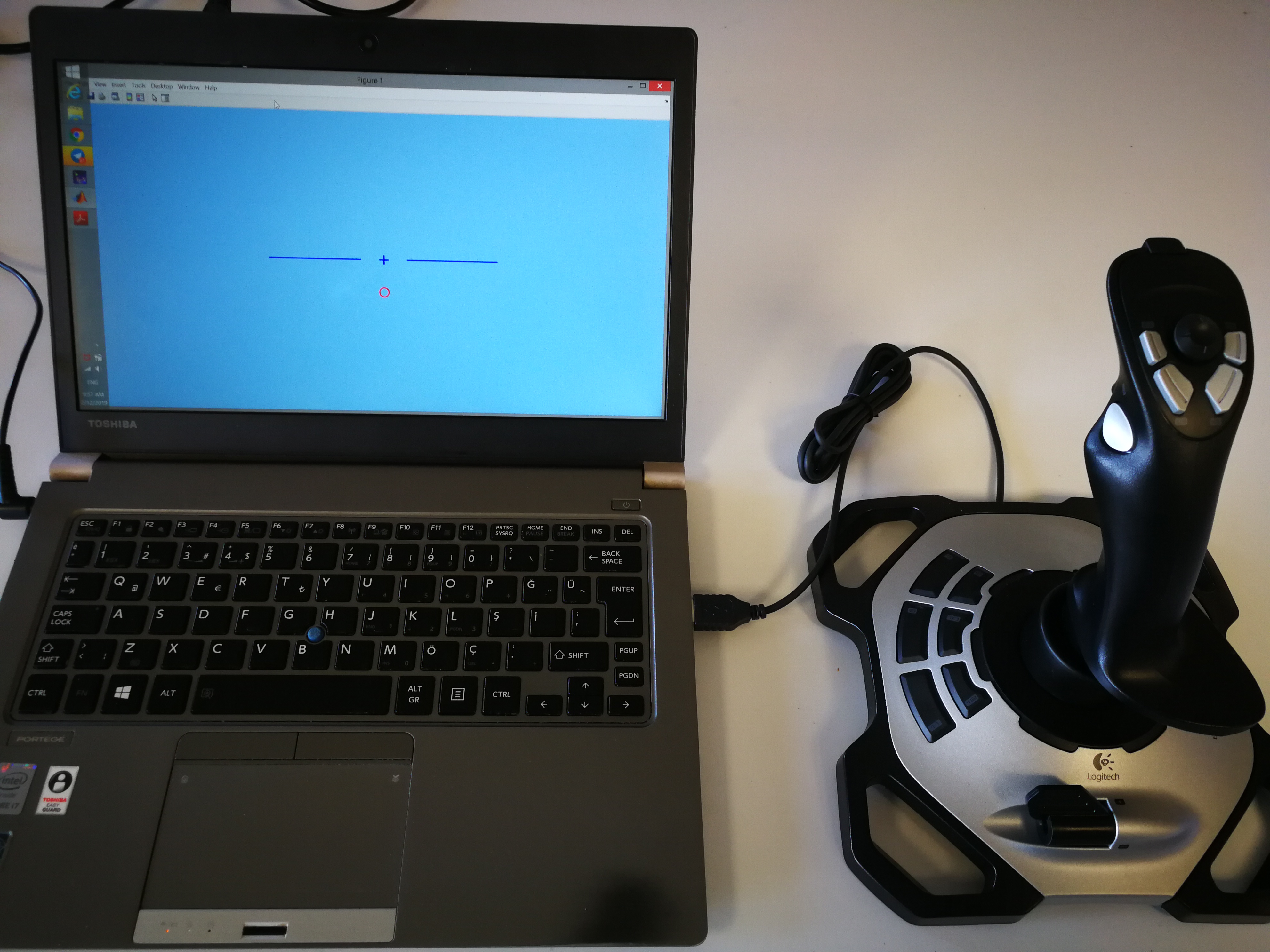}    %
		\caption{Experimental setup.} 
		\label{fig:setup}
	\end{center}
\end{figure}
\begin{figure}
	\begin{center}
		\includegraphics[width=8.5cm]{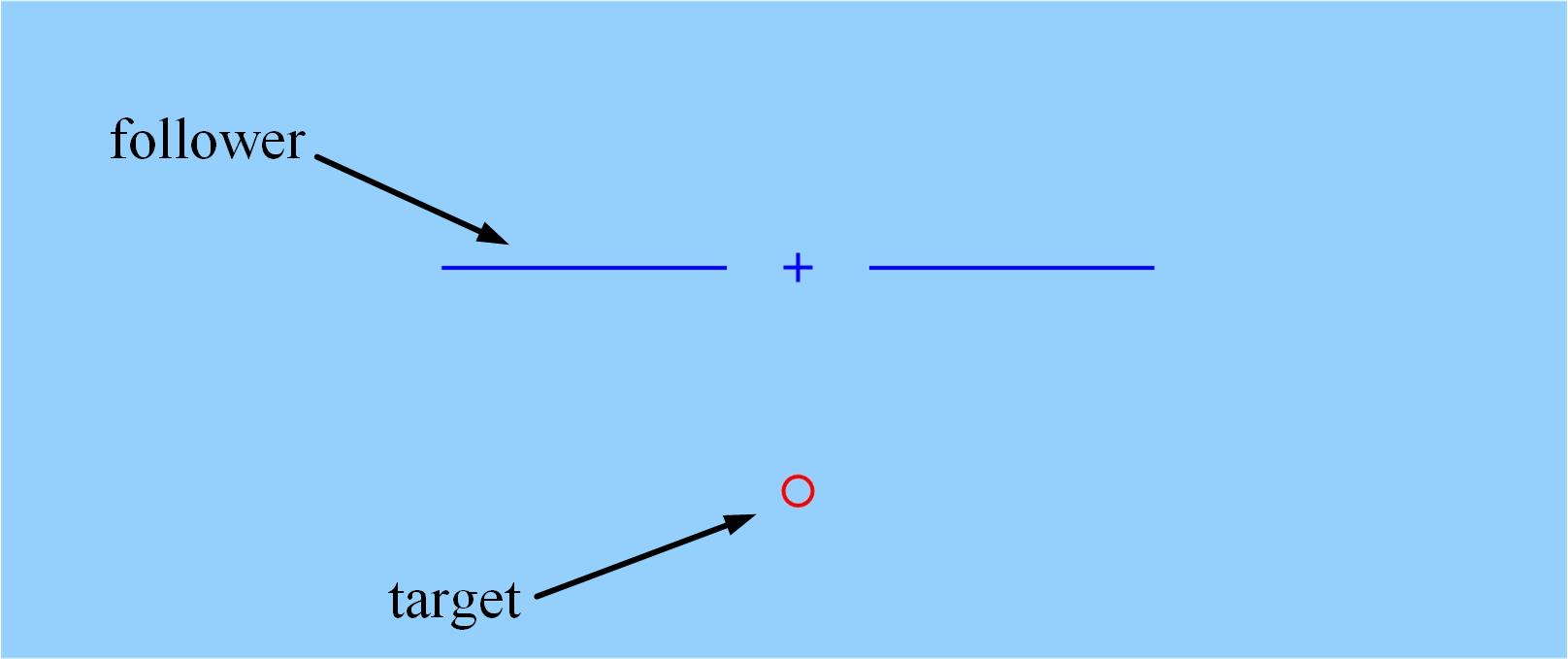}    %
		\caption{Compensatory display.} 
		\label{fig:attitudeind}
	\end{center}
\end{figure}

The reference signal $ r(t) $ is generated as a sum of the sinusoids with frequencies of 0.1, 0.3, 0.5, 0.7, 1, 1.3 and 1.5 rad/s with the same amplitude of 0.2 and without phase shift. 

Three classes of plant models, having zero, first and second order transfer functions are used in the experiments. In this section, we first give a detailed analysis of the first order plant case and then provide a summary of the results of the other cases in tables. The nominal
first order plant used in the experiments is $ Y_p(s)=\frac{4}{s+1} $, which is similar to the one used in \cite{Hess09}. The uncertainty is introduced to the plant model by modifying the gain and the pole location by $ 50\% $ to obtain $ Y_p(s)=\frac{6}{s+0.5} $. 



To form the reference model (\ref{eq:e7x}), two parameters, namely the crossover frequency and the time-delay, need to be determined. The highest frequency component of the reference signal is $ \omega_i=1.5 $ rad/s. Employing Table I for the first order plant $ Y_p $, the crossover frequency is calculated as $ \omega_c=4.5 $ rad/s. The delay is determined by using the mean value of the operators' delay, which is $ \tau=0.3 $ s. Therefore, the closed loop transfer function of the reference model is calculated as
	\begin{equation}\label{eq:e21xq}
	\begin{aligned}
	G_{cl}(s)=\frac{\frac{4.5}{s}e^{-0.3s}}{1+\frac{4.5}{s}e^{-0.3 s}}=\frac{4.5 e^{-0.3 s}}{s+4.5 e^{-0.3 s}}.
	\end{aligned} 
	\end{equation}
	
	Similar to (\ref{eq:e21xqqx}), an approximate transfer function is obtained as
	\begin{equation}\label{eq:e21xqq}
	\begin{aligned}
	\hat{G}_{cl}(s)=\frac{3.881s+24.24}{s^2+0.6834s+24.72}e^{-0.3s}.
	\end{aligned} 
	\end{equation}
	Figure \ref{fig:bode1} shows a comparison between (\ref{eq:e21xq}) and (\ref{eq:e21xqq}), and demonstrates that the approximation works well for almost all frequencies.

\begin{figure}
	\begin{center}
		\includegraphics[width=8.5cm]{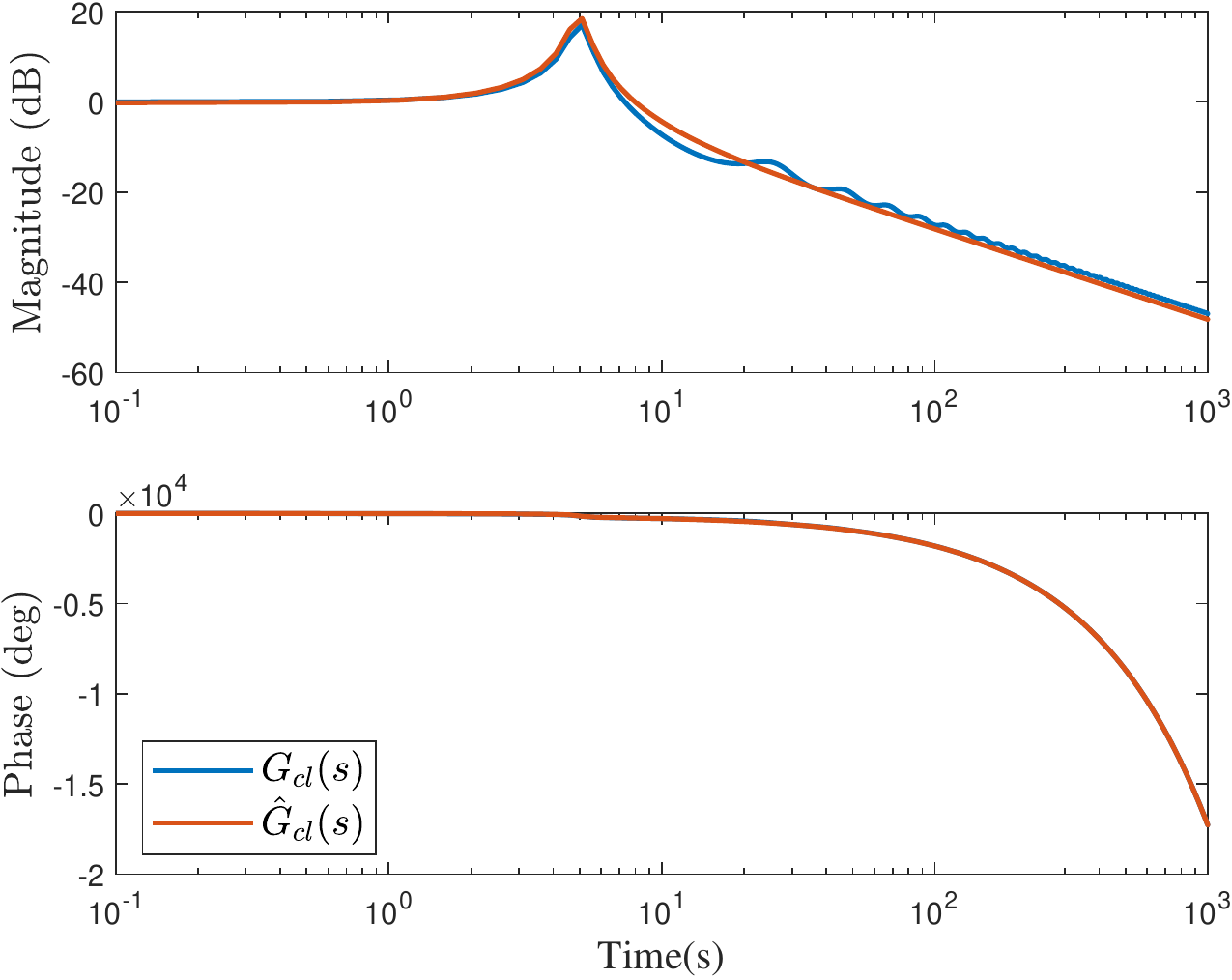}    %
		\caption{Bode plot of the reference model and its approximation.} 
		\label{fig:bode1}
	\end{center}
\end{figure}

The neuromuscular dynamics is taken as $ Y_h(s)=\frac{s+3}{s+2}e^{-0.3s} $, where the time delay $ \tau=0.3 $ is the effective time delay, including human decision making delay and neuromuscular lags.

\begin{remark}\label{rem3}
	In this paper, we assume that the neuromuscular dynamics are given. The procedure for finding the neuromuscular model can be found in \cite{Mag71, Van04}.  
\end{remark}

\subsection{The behavior of the adaptive model}

%

 The error between the plant output and the reference model is illustrated in Figure \ref{fig:error}. The effect of uncertainty injection can be seen at $ t=70 $ s. Figures \ref{fig:theta}, \ref{fig:lambda}, and \ref{fig:Kr} illustrates the adaptive human model parameters. To understand the amount of agreement between these results and the human experimental trials, visual and statistical analyses are provided in the following sections.
\begin{figure}
	\begin{center}
		\includegraphics[width=8.5cm]{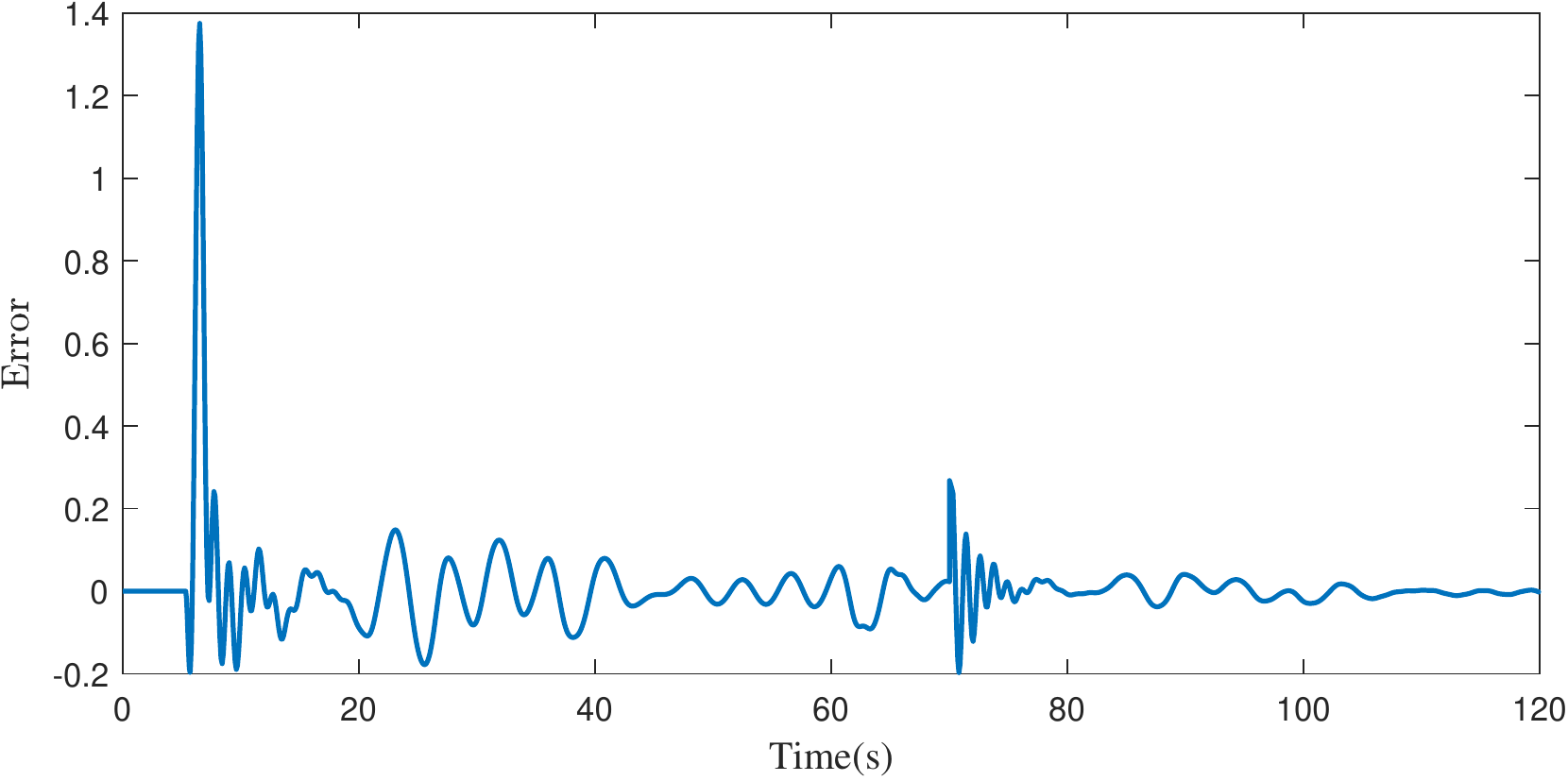}    %
		\caption{Time evolution of the error between the output of the plant controlled by the adaptive model, and the reference model output.} 
		\label{fig:error}
	\end{center}
\end{figure}
\begin{figure}
	\begin{center}
		\includegraphics[width=8.5cm]{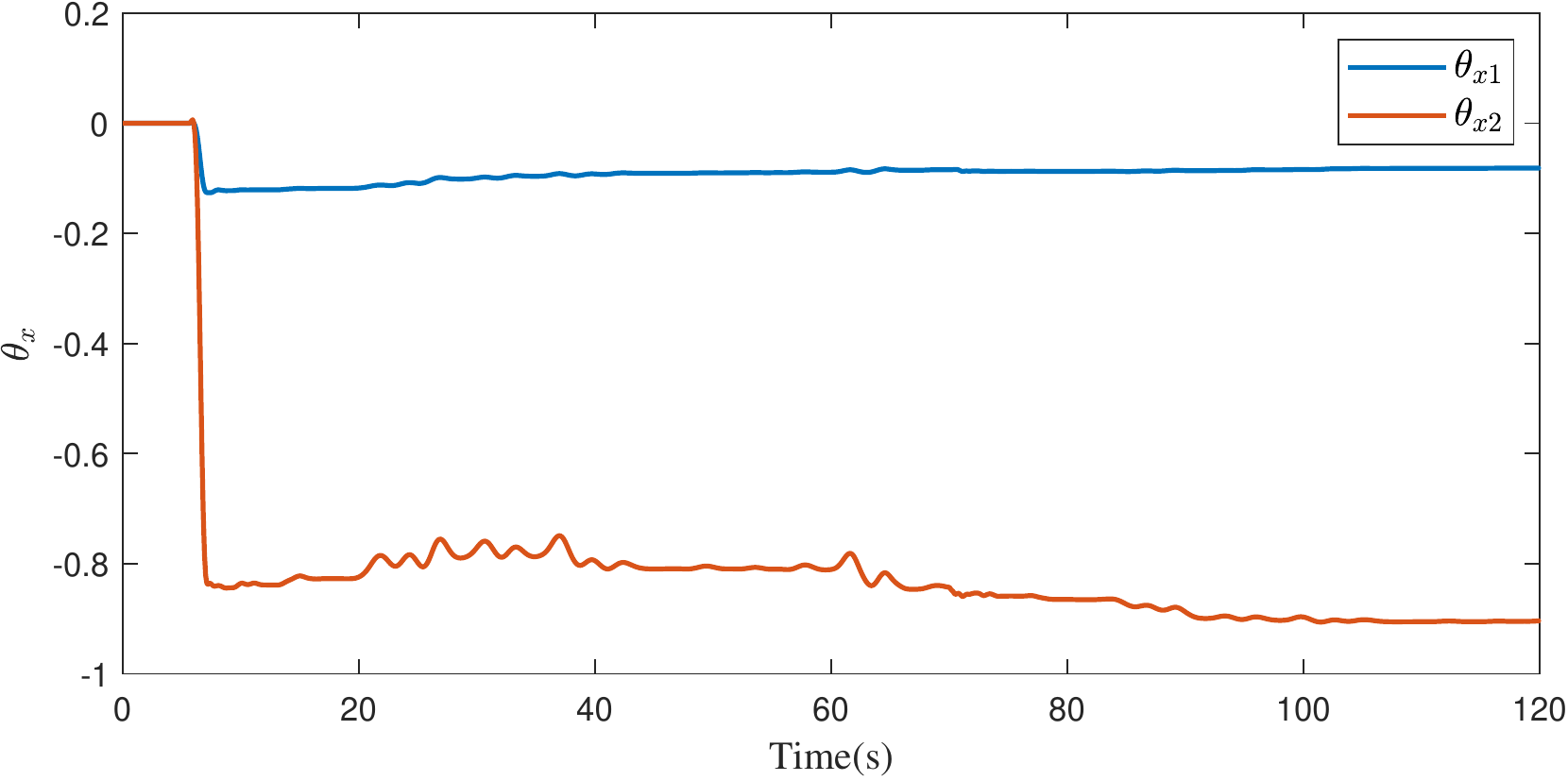}    %
		\caption{Evolution of human adaptive parameters $ \theta_{x1} $ and $ \theta_{x2} $.} 
		\label{fig:theta}
	\end{center}
\end{figure}
\begin{figure}
	\begin{center}
		\includegraphics[width=8.5cm]{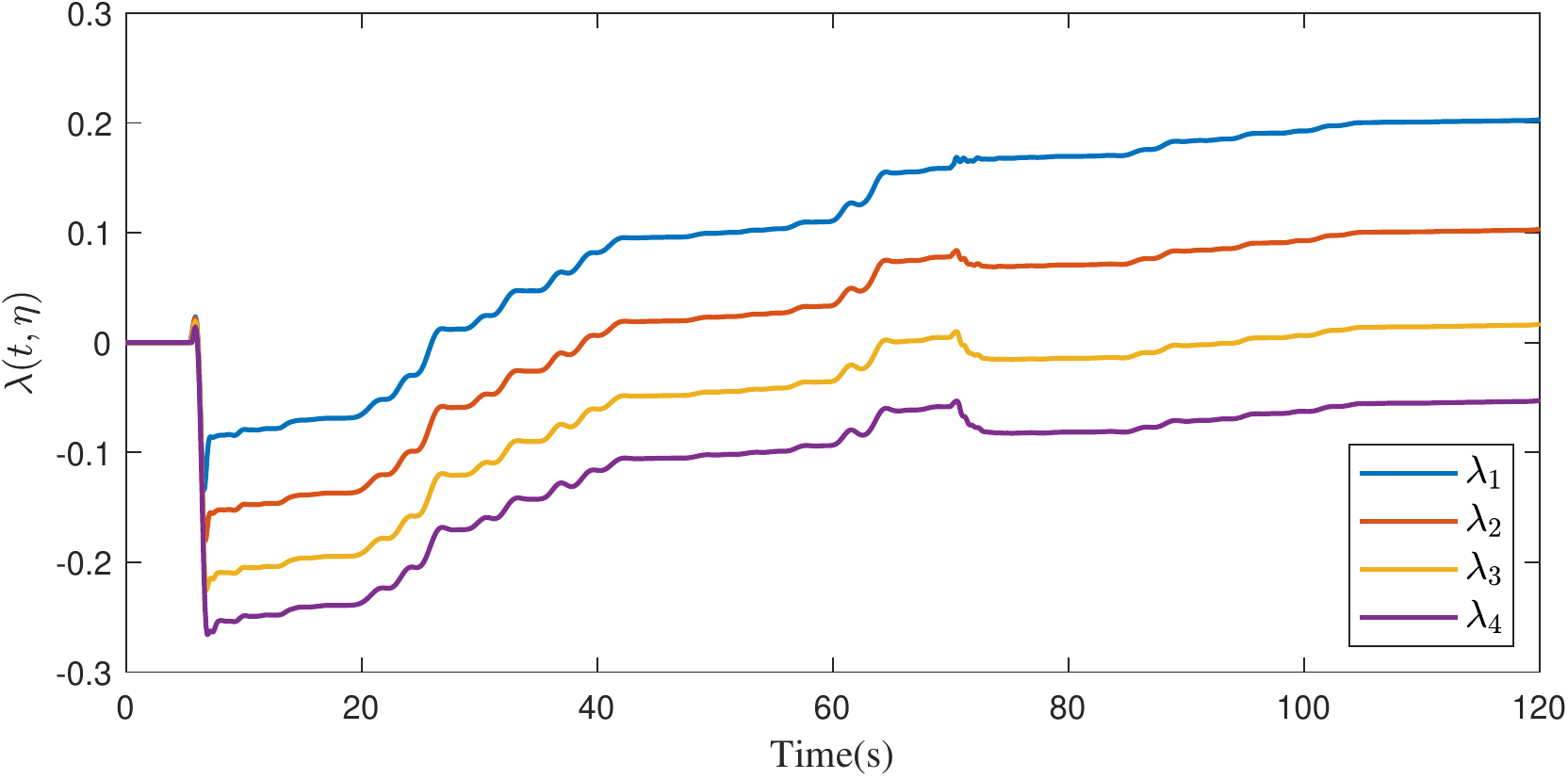}    %
		\caption{Evolution of human adaptive parameters $ \lambda_i,\ i=1, 2, 3 $ and $ 4 $.} 
		\label{fig:lambda}
	\end{center}
\end{figure}
\begin{figure}
	\begin{center}
		\includegraphics[width=8.5cm]{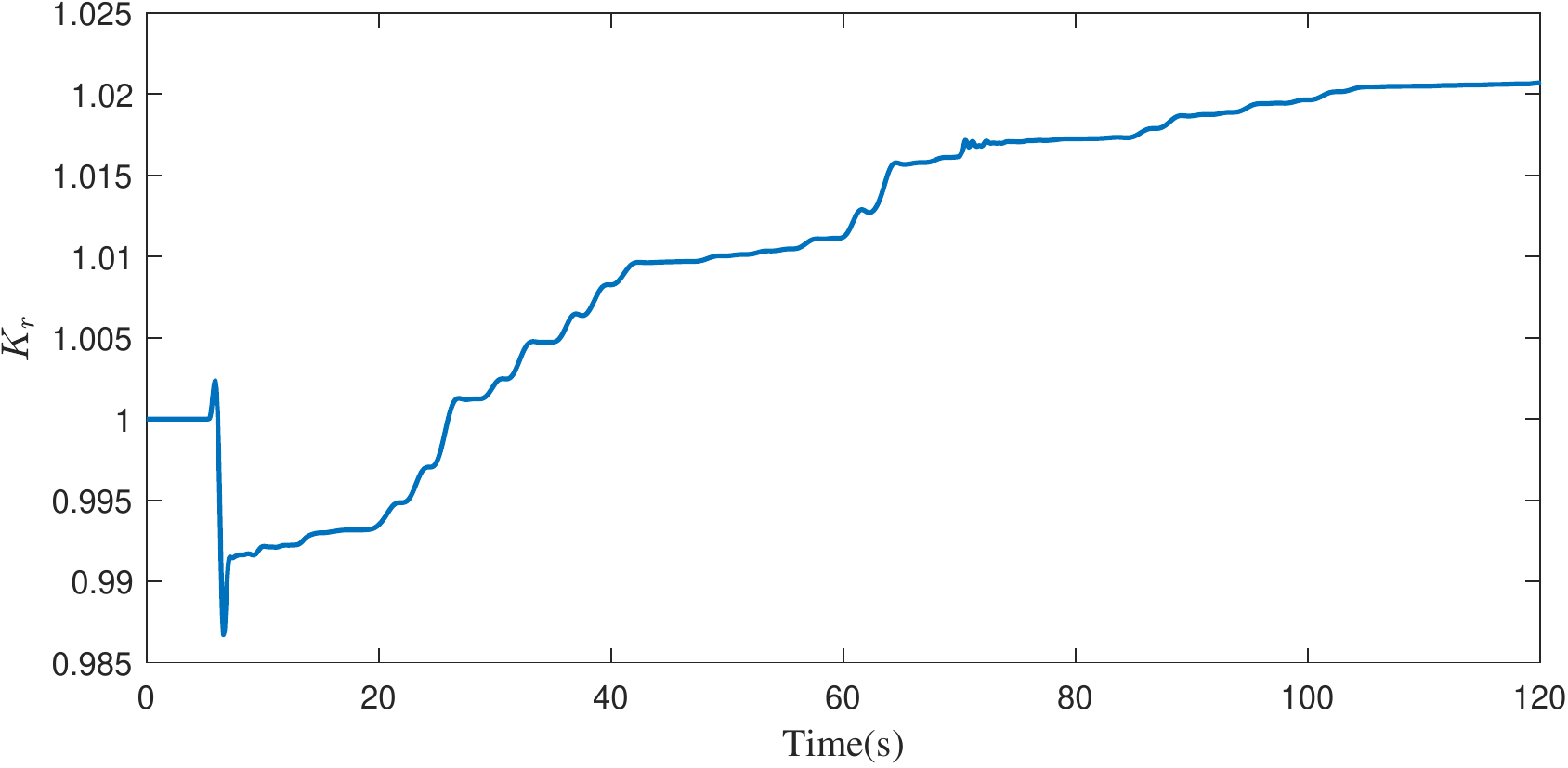}    %
		\caption{Evolution of human adaptive parameter $ K_r $.} 
		\label{fig:Kr}
	\end{center}
\end{figure}

\subsection{Participants and experimental procedure}

Eleven participants (6 women and 5 men) from the graduate and undergraduate student pools of Bilkent University participated the experiment. All of the participants read and signed the ``informed consent to participate" document. This study is approved by Bilkent University Ethics Committee for research with human participants. Before the experiments, to familiarize the participants with the experimental setup, and its environment, consisting of the display and the joystick, each participant was asked to follow a given reference via joystick inputs for the duration of $ 200 $ seconds. To prevent learning during these warm-up runs, the reference input, uncertainty injection times and the uncertainty types were chosen differently from the ones used in the real experimental runs. Specifically, the reference signal for the warm-up runs consisted of the sum of the sinusoids with frequencies of $ 0.1, 0.5, 1 $ and $ 1.5 $ rad/s with the same amplitude of $ 0.2 $ and without phase shift. The plant dynamics at the beginning of the warm-up run was $ \frac{2}{s^2+3s+2} $. At $ t = 45 $ s, the dynamics changed to $ \frac{5}{s+2} $ in a step like manner (suddenly). It changed to $ \frac{3}{s+1} $ at around $ t = 90 $ s using a sigmoid function (gradually), and again changed to a zero order dynamics at $ 150 $ s (suddenly).

\subsection{A visual analysis of the adaptive model}

Let $ f_{p1}(t), f_{p2}(t), ..., f_{pk}(t) $ be the plant outputs when participants $ p1, p2, ..., pk $ are in the loop, respectively. For each $ f_{pi}(t) $, $ t=T_1, T_2, ..., T_N $, where $ T_j, j=1, 2, ..., N $, represents a sampling instant. At each sampling instant $ T_j $, the minimum, the maximum and the mean values of the plant outputs when participants are in the loop can be obtained as
\begin{align}
f_{p_{min}}(T_j)&=\min_{i=1, 2, ..., k}f_{pi}(T_j),\ \ j=1, ..., N,\label{eq:estat4} \\
f_{p_{max}}(T_j)&=\max_{i=1, 2, ..., k}f_{pi}(T_j),\ \ j=1, ..., N, \label{eq:estat5} \\
f_{p_{mean}}(T_j)&=\frac{\sum_{i=1}^{k}f_{pi}(T_j)}{k}, \ \ j=1, ..., N,
\end{align} 
 where $ k=11 $ is the number of participants. Figure \ref{fig:minmaxadapt} shows the evolutions of  $ f_{p_{min}} $ and $ f_{p_{max}} $, together with $ f_{ad}(t)\in \mathbb{R}^N $, which is the plant output when adaptive human model is in the loop, where $ t= T_1, T_2, ..., T_N $. It is seen that the plant output when adaptive human model is in the loop almost always stays between the maximum and the minimum values of the plant output when participants are in the loop. Furthermore, Figure \ref{fig:meanadapt} demonstrates that $ f_{p_{mean}} $ and $ f_{ad} $ evolve reasonably close to each other.
  

\begin{figure}
	\begin{center}
		\includegraphics[width=9.0cm]{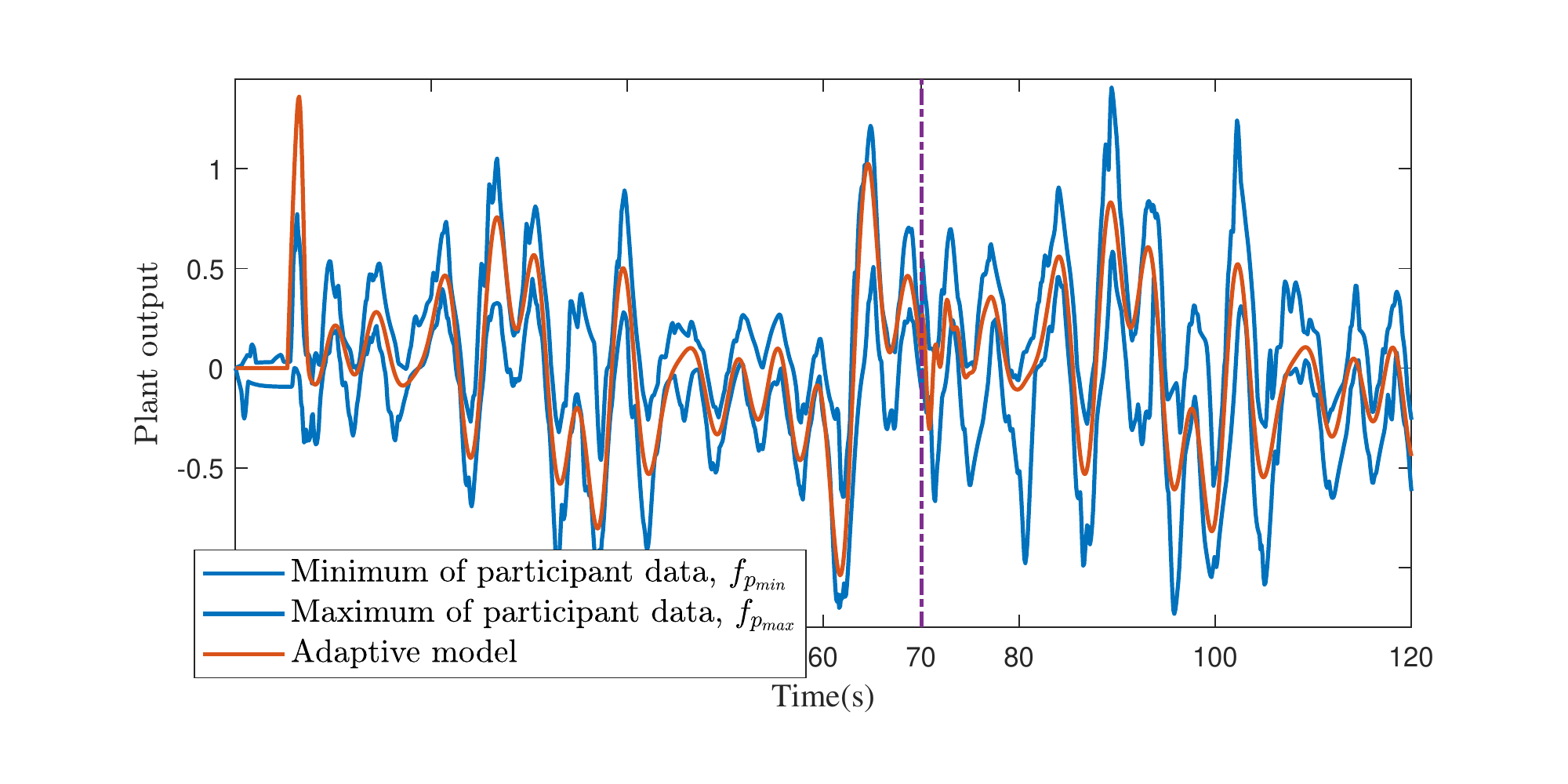}    %
		\caption{Plant output, $ x_{hp} $, when adaptive human model is in the loop vs. minimum and maximum values of plant output when participants are in the loop.} 
		\label{fig:minmaxadapt}
	\end{center}
\end{figure}

\begin{figure}
	\begin{center}
		\includegraphics[width=9.0cm]{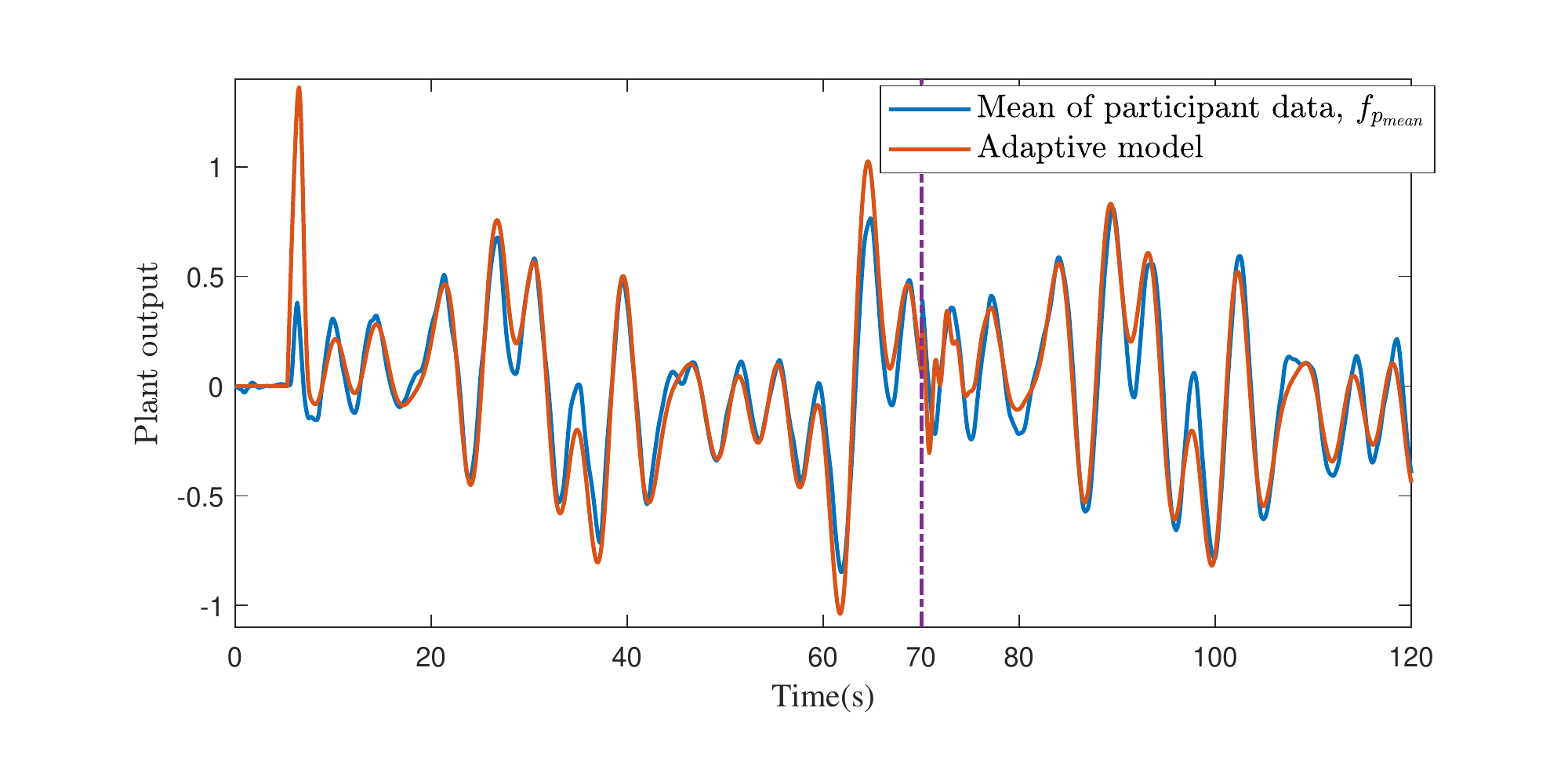}    %
		\caption{Plant output, $ x_{hp} $, when adaptive human model is in the loop vs. mean value of plant output when participants are in the loop.} 
		\label{fig:meanadapt}
	\end{center}
\end{figure}

%

\subsection{Statistical analysis of the adaptive model using confidence intervals}

The difference between the plant output when the $ i^{\text{th}} $ participant is in the loop and when the adaptive human model is in the loop is defined as
\begin{equation}\label{eq:estat1}
\begin{aligned}
d_i\equiv f_{ad}-f_{pi}, \ \ i=1, ..., k,
\end{aligned} 
\end{equation}
where $ d_i=[d_i(T_1), ..., d_i(T_N)]^T\in \mathbb{R}^N, i=1, ..., k $, is called the $ i^{\text{th}} $ difference. The mean and the standard deviation of the $ i $th difference is obtained as
\begin{align}
\bar{d}_i&= \frac{\sum_{j=1}^{N}d_i(T_j)}{N}, \ \ i=1, ..., k,\label{eq:estat2} \\
s_i&=\sqrt{\frac{\sum_{j=1}^{N}(d_i(T_j)-\bar{d}_i)^2}{N-1}}, \ \ i=1, ..., k.\label{eq:estat3}
\end{align}

\begin{figure}
	\begin{center}
		\includegraphics[width=9.0cm]{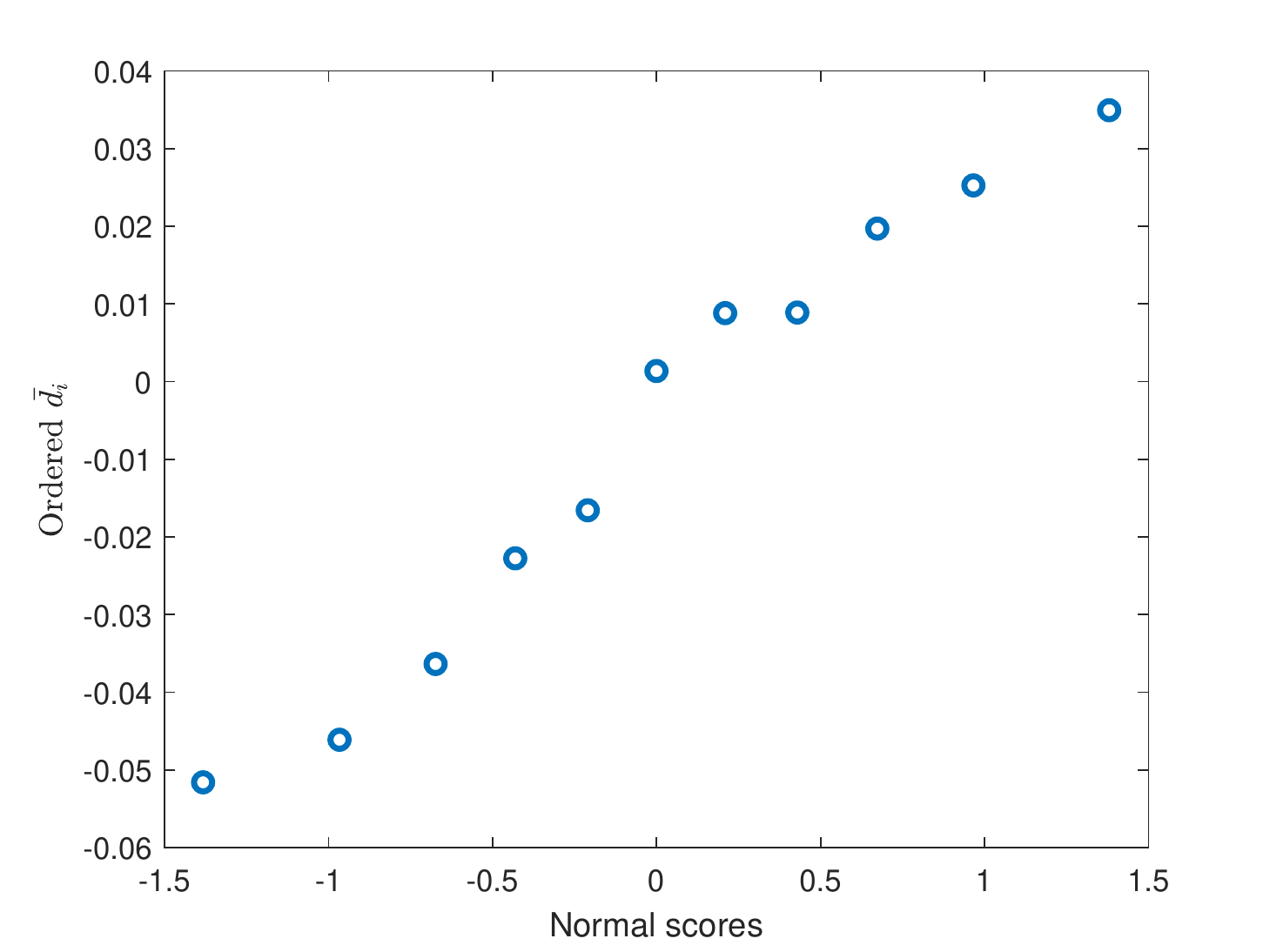}    %
		\caption{Normal-scores plot} 
		\label{fig:normalscores}
	\end{center}
\end{figure}

The normal-scores plot for $ \bar{d}_i $ is given in Figure \ref{fig:normalscores}. The figure does not show any significant deviation from the normal distribution. This shows us that the data do not suggest that the \textit{population} of mean-errors, $ \bar{d}_i $, deviates significantly from normal distribution. The \textit{sample} mean and the \textit{sample} standard deviation of $ \bar{d}_i $'s can be obtained as
\begin{align}
\bar{d}&=\frac{\sum_{i=1}^{k}\bar{d}_i}{k}, \label{eq:estat6}\\
s&=\sqrt{\frac{\sum_{i=1}^{k}(\bar{d}_i-\bar{d})^2}{k-1}}. \label{eq:estat7}
\end{align}
Let $ \mu_0 $ be the mean value of the \textit{population} of mean-errors, given as
\begin{align}\label{eq:estat8}
\mu_0\equiv \frac{\sum_{i=1}^{K}\bar{d}_i}{K},
\end{align}
where $ K $ is the \textit{population} size. Since normal-scores plot, given in Figure \ref{fig:normalscores}, didn't provide any counter evidence, assuming that the distribution of the set of data $ \{\bar{d}_1, ..., \bar{d}_K\} $ is normal with mean $ \mu_0 $, $ \mu_0 $ satisfies the following probability \cite{Joh19}
\begin{align}\label{eq:estat9}
P\left[\bar{d}-t_{\alpha/2}\frac{s}{\sqrt{k}}<\mu_0<\bar{d}+t_{\alpha/2}\frac{s}{\sqrt{k}} \right]=1-\alpha,
\end{align}
where $ \bar{d} $ and $ s $ are obtained from (\ref{eq:estat6}-\ref{eq:estat7}), $ k $ is the number of participants, $ \alpha $ is the significance level, and $ t_{\alpha/2} $ is the upper $ \alpha/2 $ point of the $ t $ distribution with degree of freedom $ k-1 $, which can be obtained from the $ t $-distribution table. Since the number of participants, $ k=11 $, is less than $ 30 $, it is appropriate to use the $ t $-distribution.
Using $ \alpha=0.05 $, obtaining $ t_{\alpha/2} $ from the $ t $-distribution table as $ 2.228 $, and calculating $ \bar{d} $ as $ -0.0068 $ and $ s $ as $ 0.0379 $, it can be concluded using (\ref{eq:estat9}) that we are $ 95\% $ confident that $ \mu_0 $ is in the interval $ (-0.0323, 0.0187) $. This shows that the mean $ \mu_0 $ of the \textit{population}'s mean deviation from the adaptive human model is reasonably close to zero. 



Similarly, the variance, $ \sigma_0^2 $, of the \textit{population}'s mean deviation from the adaptive human model satisfies the following probability 
\cite{Joh19}
\begin{align}\label{eq:estat11}
P\left[\frac{(k-1)s^2}{\chi_{\alpha/2}^2}<\sigma_0^2<\frac{(k-1)s^2}{\chi_{1-\alpha/2}^2} \right]=1-\alpha,
\end{align}
where $ \chi_{\alpha/2}^2 $ is the upper $ {\alpha/2} $ point of the $ \chi^2 $ distribution with degree of freedom $ k-1 $ and can be obtained from the $ \chi^2 $ distribution table. Calculating $ s $ from (\ref{eq:estat7}), using $ \alpha = 0.05 $, and obtaining $ \chi_{\alpha/2}^2 $ and $ \chi_{1-\alpha/2}^2 $ from the $ \chi^2 $ table with $ 10 $ degrees of freedom, it can be concluded using (\ref{eq:estat11}) that we are $ 95\% $ confident that $ \sigma_0 $ is in the interval $ (0.0265,\ 0.0663) $. This shows that the standard deviation $ \sigma_0 $ of the \textit{population}'s mean deviation from the adaptive human model is reasonably small. 


\subsection{Statistical analysis of the adaptive model using hypothesis testing}

In this analysis, we test whether the hypothesis ``the mean value of the \textit{population} mean-errors, or the mean deviations from the adaptive model," is zero. In other words, our null hypothesis, $ H_0 $, is given as
\begin{align}
H_0: \mu_0=0,
\end{align}
where $ \mu_0 $ is defined in (\ref{eq:estat8}). The alternative hypothesis, $ H_1 $, is given as $ H_1: \mu\neq 0 $. Similar to the confidence interval analysis, assuming that $ \mu_0 $ is the mean of a normally distributed set of data $ \{\bar{d}_1, ..., \bar{d}_K\} $ where $ K $ is the \textit{population} size, the hypothesis $ H_0 $ is rejected if,
\begin{align}\label{eq:estat10}
\left| \frac{(\bar{d}-\mu_0)\sqrt{k}}{s}\right| \geq t_{\alpha/2},
\end{align}
where $ \bar{d} $ and $ s $ are obtained from (\ref{eq:estat6}-\ref{eq:estat7}), $ k $ is the number of participants and $ t_{\alpha/2} $ is the upper $ \alpha/2 $ point of the $ t $ distribution with degree of freedom $ k-1 $ \cite{Joh19}. Using the significance level $ \alpha=0.05 $ and degree of freedom $ k-1=10 $, obtaining $ t_{0.025}=2.228 $ from the $ t $-distribution table, calculating $ \bar{d}=-0.0068 $ and $ s=0.038 $ using (\ref{eq:estat6}) and (\ref{eq:estat7}), respectively, and substituting $ \mu_0=0 $ and $ k=11 $, the left hand side of (\ref{eq:estat10}) can be calculated as $ 0.5935 $, which is less than $ t_{\alpha/2} $. Therefore, we cannot reject $ H_0 $. We retain $ H_0 $ and conclude that $ H_1 $ fails to be proved.
 
Since we are retaining the null hypothesis, we want to minimize the probability $ \beta $ of incorrectly retaining the null hypothesis. This means that we want our test's power, $ 1-\beta $, to be large, such as $ 0.95 $. What is the minimum required deviation of the \textit{population} mean from $ 0 $, represented as $ \mu_1 $, that would make our test to incorrectly retain the null hypothesis with $ 0.05 $ probability, i.e. $ \beta=0.05 $? To calculate this, we first write the rejection region, $ R $, using (\ref{eq:estat10}) as
\begin{align}
R:\ \left| \frac{(\bar{d}-\mu_0)\sqrt{k}}{s}\right| \geq t_{\alpha/2}\ 
 \implies R:\ |\bar{d}|\geq 0.0255.
\end{align}
Defining $ T=\frac{(\bar{d}-\mu_1)\sqrt{k}}{s} $, for $ \beta=0.05 $, we need
\begin{align}\label{eq:estat12}
&P\left[ \frac{(-0.0255-\mu_1)\sqrt{k}}{s}<T<\frac{(0.0255-\mu_1)\sqrt{k}}{s} \right]  \notag \\
&=\frac{\beta}{2}=0.025.
\end{align}
Using the $ t $-table, it can be found that the minimum $ |\mu_1| $ that satisfies (\ref{eq:estat12}) is $ 0.051 $. This means that our test can detect an $ 0.051 $ deviation from the mean value of the mean-error between the adaptive human model and the participant data when the probability of the test to incorrectly conclude that the model and the data are compatible ($ \mu_0=0 $) is only $ 5\% $. 

\begin{table}\label{table2}
	\caption{Sudden uncertainty} 
	\centering
	\begin{tabular}{ | c | c | c | c | }
		\hline 
		& 0 order & 1st order & 2nd order \\
		\hline 
		TF before $ 70 $ s & $ 4 $ & $ \frac{4}{s+1} $ & $ \frac{4}{(s+1)(s+5)} $ \\
		\hline
		TF after $ 70 $ s & $ 6 $ & $ \frac{6}{s+0.5} $ & $ \frac{6}{(s+0.5)(s+2.5)} $ \\
		\hline
		$ \bar{d} $ & $ 0.0085 $ & $ -0.0068 $ & $ 0.0011 $ \\
		\hline 
		$ s $ & $ 0.0339 $ & $ 0.0379 $ & $ 0.0252 $ \\
		\hline
		Mean conf. int. & $ (-0.014, 0.03) $ & $ (-0.032, 0.02) $ & $ (-0.016, 0.018) $ \\
		\hline
		St.d. conf. int. & $ (0.024, 0.06) $ &  $ (0.026, 0.066) $ & $ (0.018, 0.044) $ \\
		\hline
		Hypothesis test & $ H_0 $ is retained & $ H_0 $ is retained  & $ H_0 $ is retained \\
		\hline
	\end{tabular}
\end{table} 
\begin{table}\label{table3}
	\caption{Gradual uncertainty} 
	\centering
	\begin{tabular}{ | c | c | c | c | }
		\hline 
		& 0 order & 1st order & 2nd order \\
		\hline 
		TF before $ 70 $ s & $ 4 $ & $ \frac{4}{s+1} $ & $ \frac{4}{(s+1)(s+5)} $ \\
		\hline
		TF after $ 70 $ s & $ 6 $ & $ \frac{6}{s+0.5} $ & $ \frac{6}{(s+0.5)(s+2.5)} $ \\
		\hline
		$ \bar{d} $ & $ 0.0154 $ & $ 0.0026 $ & $ 0.004 $ \\
		\hline 
		$ s $ & $ 0.038 $ & $ 0.034 $ & $ 0.03 $ \\
		\hline
		Mean conf. int. & $ (-0.01, 0.04) $ & $ (-0.02, 0.025) $ & $ (-0.016, 0.023) $ \\
		\hline
		St.d. conf. int. & $ (0.026, 0.067) $ & $ (0.0235, 0.06) $ & $ (0.02, 0.05) $ \\
		\hline
		Hypothesis test & $ H_0 $ is retained & $ H_0 $ is retained  & $ H_0 $ is retained \\
		\hline
	\end{tabular}
\end{table} 


Analyses of the experimental results where a first order plant dynamics is used with a sudden uncertainty injection is provided above. All of the results, including the ones for the other cases, where plants with different orders and sudden/gradual uncertainty injections, are summarized in Tables II and III.
The data collected from the participants can be reached at \href{http://www.syslab.bilkent.edu.tr/research}{http://www.syslab.bilkent.edu.tr/research}.
%
%
%
%

\section{SUMMARY}\label{sec:conclusion}

In this paper, an adaptive human pilot model based on model reference adaptive control principles is proposed. This model mimics the pilot decision making process by making sure that the overall closed loop system follows the crossover model in the presence of plant uncertainties. The stability of the system is shown using the Lyapunov-Krasovskii stability criteria. Furthermore, experiments with human operators are conducted to validate the model. Detailed visual and statistical analyses of the experimental results show that the adaptive model creates similar system responses as the human operators. 

\bibliographystyle{ieeetr}
\bibliography{References}

\end{document}